%% file: Path_Decomposition.tex
\documentclass[11pt]{article}

\usepackage{authblk}

\title{Reconciling Selfish Routing with Social Good}

\author[]{Soumya Basu}
\author[]{Ger Yang}
\author[]{Thanasis Lianeas}
\author[]{Evdokia Nikolova}
\author[]{Yitao Chen}
\affil[]{Department of Electrical and Computer Engineering, \\ The University of Texas at Austin}

\date{\vspace{-5ex}}

\usepackage{algorithm}
\usepackage{algorithmic}
\usepackage{amsmath,amssymb, amsthm}
\usepackage{color}
\usepackage{graphicx}
\usepackage{subcaption}
\usepackage{epstopdf}
\usepackage{bbm}
\usepackage{bm}
\usepackage[makeroom]{cancel}
\usepackage[margin=1in]{geometry} 
\usepackage[titletoc,title]{appendix}

\newcommand{\mc}[1]{\mathcal{#1}}

\newcommand{\pathdecomp}{path flow}
\newcommand{\pathdecomps}{path flows}

\newcommand*{\Scale}[2][4]{\scalebox{#1}{$#2$}}%
\newcommand{\argmin}{\mathrm{argmin}}
\newcommand{\argmax}{\mathrm{argmax}}
\newcommand{\efu}{edge flow unfairness}

\graphicspath {{fig/}}

\theoremstyle{definition}
\newtheorem{theorem}{Theorem}
\theoremstyle{definition}
\newtheorem{lemma}{Lemma}
\theoremstyle{definition}
\newtheorem{proposition}{Proposition}
\theoremstyle{definition}
\newtheorem{claim}{Claim}
\theoremstyle{definition}
\newtheorem{corollary}{Corollary}
\theoremstyle{definition}
\newtheorem{definition}{Definition}
\theoremstyle{remark}
\newtheorem*{example}{Example}
\theoremstyle{remark}
\newtheorem*{remark}{Remark}

\usepackage{soul,xcolor}

\definecolor{purple}{rgb}{0.7, 0, 0.6}
\definecolor{dgreen}{rgb}{0, 0.8, 0}

\newcommand{\evn}[1]{\textcolor{black}{#1}}

\begin{document}

\maketitle
\begin{abstract}
\input{abstract}

\end{abstract}

\input{introv4.tex}

\input{prelimv1.tex}
\input{solutionconcepts.tex}

\input{hier.tex}
\input{cost.tex}

\input{complexv2.tex}
\input{design.tex}

\input{conclusion.tex}

\input{acknowledgements.tex}

\bibliographystyle{unsrt} 
\bibliography{mybib} 

\end{document}

%% file: abstract.tex
Selfish routing is one of the 
most studied problems in algorithmic game theory, with one of the principal applications being that of routing in road networks. The majority of related work, in the many variants of the problem,  deals with the inefficiency of equilibria to which users are assumed to converge.
Multiple mechanisms for improving the outcomes at equilibria have been considered, such as the use of tolls or the use of Stackelberg strategies, each with different caveats in terms of their applicability to real traffic routing.  
But the emergence of routing technologies and autonomous driving motivates new solution concepts that can be considered as outcomes of the game and may help in improving the network's performance. In reality, when users ask their routing devices for good origin to destination paths, they care about the end-to-end delay on their paths without (directly) caring about subpath optimality. This gives a central planner the ability, through routing devices, to provide path flow solutions that circumvent the local subpath optimality conditions imposed by (approximate) Nash equilibria, while they are acceptable to the players and potentially have good social cost.

Inspired by the above observation, we consider three possible outcomes for the game: (i) $\theta$-Positive Nash Equilibrium flow,  where every path that has non zero flow on all of its edges has cost no greater than $\theta$ times the cost of any other path, (ii) $\theta$-Used Nash Equilibrium flow, where every path that appears in the path flow decomposition has cost no greater than $\theta$ times the cost of any other path, and (iii) $\theta$-Envy Free flow, where every path that appears in the path flow decomposition has cost no greater than $\theta$ times the cost of any other path in the path flow decomposition.
We first examine the relations of these outcomes among each other and then measure their possible impact on the network's performance, through the notions of price of anarchy and price of stability. Afterwards, we examine the computational complexity of finding such flows of minimum social cost and give a range for $\theta$ for which this task is easy and a range for $\theta$ for which, for the newly introduced concepts of $\theta$-Used Nash Equilibrium flow and $\theta$-Envy Free flow, this task is NP-hard. Finally, we propose deterministic strategies which, in a worst case approach,  can be used by a central planner in order to provide good such flows, and also introduce a natural idea for randomly routing players after giving them specific guarantees about their costs in the randomized routing, as a tool for the central planner to implement a desired flow.

%% file: introv4.tex
\section{Introduction}\label{sec:intro}

\subsection{Two sides of the coin: Social Welfare vs Selfishness} \label{subSect:SOFandFairF}
A fundamental problem arising in the management of road-traffic and communication networks is routing traffic to optimize network performance. In the setting of road-traffic networks the average delay incurred by a unit of flow quantifies the cost of a routing assignment.  From a collective perspective minimizing the average cost translates to maximizing the welfare obtained by society.  
Starting from the seminal works of Wardrop \cite{wardrop1952some} and Beckman et al. \cite{beckmann1956studies}, the literature on network games has differentiated between 1) the objective of a central planner to minimize average cost and thus find a socially optimal (SO) flow, and 2) the selfish objectives of users minimizing their respective costs.  In the latter case,  
the network users acting in their own interest are assumed to converge to a Nash Equilibrium (NE) flow as further rerouting fails to improve their own objective.  

The tension between the central planner and individual objectives has been an object of intense study in the 
algorithmic game theory literature on congestion games. A central question arising in congestion games, ``how much does network performance suffer from selfish behavior?'', has been investigated extensively through the notions of Price of Anarchy (PoA) and Price of Stability (PoS), namely the ratio of the maximum cost among all Nash equilibria over the social optimum and the ratio of the minimum cost among all  Nash equilibria over the social optimum, respectively. Prior research shows that the Nash equilibrium flow may attain very poor social welfare compared to social optimum, i.e. we may get only poor bounds on the PoA or the PoS, with these bounds being tight for some classes of instances.  For an overview we refer the reader to the survey~\cite{roughgarden2002selfish}. 

This discrepancy between selfishness and social good calls for finding a middle ground between the two ends of the spectrum---the Nash flow and the socially optimal flow.  
To that end, previous research on mechanism design
has lead to theoretically appealing solutions such as 
toll placement and Stackelberg routing~\cite{swamy2012effectiveness}. Placing tolls on edges has been shown to improve the network performance up to the point of completely optimizing it when there are no restrictions on the tolls' values. Using Stackelberg strategies, where one assumes that a fraction of users is willing to cooperate and follow the routes dictated by the central planner, has also been theoretically shown to improve the network performance. In spite of the nice properties of these solutions that induce selfish users to act in a socially friendly way, these mechanisms have faced criticism in the real world in terms of their implementation and their fairness towards various users. 
 
To mitigate the tension between selfishness and social good in a way that is more fair to the users, we set out to explore the properties of alternative solution concepts 
where users under some reasonable incentive condition 
adopt a ``socially desirable" routing of traffic in between the Nash equilibrium (which has high social cost) and the social optimum (which may be undesirable/unfair to users on the longer paths) \cite{roughgarden2002unfair}.  
The advent of routing applications and the growing dependence of users on these applications 
places us at an epoch when such new ideas in mechanism design may be more relevant and also more readily integrated to practice. 
Consider the scenario where some routing application presents the uninformed users with routes alongside the guarantees of ``relative fairness'' and ``reasonable delay'' and the users adopt the paths. This scenario is close to reality, since users unaware of the network congestion often use their routing devices to travel to their destinations or pick a path that has been presented to them before. 

This naturally brings forth the questions of whether there exist solutions (flows) where good social welfare is achieved under an appropriate incentive condition for the users 
and if such solutions can be efficiently computed. 
An example of such a solution could be enforcing a $\theta$-approximate Nash equilibrium of low social cost,  where users are guaranteed to get assigned a path of cost no greater than $\theta$ times the cost of the shortest path and as such, the solution is ``relatively fair".\evn{\footnote{We note that the concept of fairness has been considered in the literature of routing games in more than one ways. The two main approaches define fairness as: 1) the ratio of the maximum path delay in a given flow to the average delay under Nash equilibrium~\cite{roughgarden2002unfair} and 2) the ratio of the maximum path delay to the minimum path delay in a given flow~\cite{correa2007fast}.}} 
Yet, other solution concepts seem to arise naturally and are introduced below.



\subsection{Selfishness and Envy}
In our quest to achieve the coveted middle ground between the social optimum and Nash equilibrium, by combining good social welfare with satisfied users, we present two notions related to: 1) Selfishness and 2) Envy. 

Firstly, we consider selfishness where users tend to selfishly improve their own cost whenever there exists some scope of improvement.  This conforms to the notion of Nash Equilibrium and a slight relaxation of  absolute selfishness leads us to the approximate Nash Equilibrium concept.  Specifically, we consider a multiplicative approximation consistent with the tradition in approximation algorithms: we refer to a $\theta$-Nash equilibrium flow as a flow in which the length of any {\em used} path in the network is less than or equal to $\theta$ times the length of any other path in the network, with $\theta \geq1$.   Note that for $\theta=1$ we have the Nash Equilibrium flow.

The existing literature in congestion games mainly regards a used path as a path that has positive flow in all of its edges, independent of the path flow decomposition that induces the edge flow. Here we make the distinction between \emph{positive} paths, i.e. paths that have positive flow in all of their edges, (note, this is independent of the path flow decomposition) 
 and \emph{used} paths, i.e. paths that appear in the path flow decomposition with positive flow. With these definitions we define a $\theta$-Positive Nash Equilibrium ($\theta$-PNE)\footnote{We remark that in the literature, PNE is typically used for abbreviating Pure Nash Equilibrium.  In this paper, we always use it to mean Positive Nash Equilibrium as we define it here.} 
 to be a flow in which the length of any {\em positive} path in the network is less than or equal to $\theta$ times the length of any other path, and  a $\theta$-Used Nash Equilibrium ($\theta$-UNE) to be a flow in which the length of any {\em used} path in the network is less than or equal to $\theta$ times the length of any other path. 
Specifically, the concept of UNE deals directly with the paths assigned to users whereas PNE deals with positive paths which may remain unused.  
 As we shall see, the set of $\theta$-PNE flows is a subset of the set of $\theta$-UNE flows and this inclusion might be strict, though for $\theta=1$ these sets coincide. The definition of $\theta$-approximate Nash equilibrium in the literature~\cite{christodoulou2011performance} corresponds to that of $\theta$-UNE. However, to the best of our knowledge, the significance of path flows in the definition of $\theta$-UNE has not been made explicit in any prior work.
 
Next, consider the notion of envy where for the same source and destination a user experiences envy against another user if the latter incurs smaller delay compared to the former under a given path flow. Similarly to the approximate Nash equilibrium flow we can consider a notion of approximately envy free flows where in a $\theta$-Envy Free ($\theta$-EF) flow, the ratio of any two {\em used} paths in the network is upper bounded by $\theta$, for some $\theta\geq 1$.  Note, the difference from the $\theta$-UNE definition is that a used path's cost is compared only to other used paths' costs. 
Envy free flows arise naturally as we consider the routing applications setup where users only collect information about the routes provided by the application. Thus, on the one hand, the possible costs for the current users in some sense compare to the costs of the users that have already used the network. On the other hand, routes for which there is no (sufficient) information potentially may never appear as an option. In other words, routes that have not been chosen in the past (sufficiently many times), i.e. ``unused routes'', do not arise in the comparison of the paths' costs. 

An example of how the concepts of $\theta$-PNE, $\theta$-UNE, and $\theta$-EF may differ from each other is illustrated in Figure~\ref{fig:ex_flows}. There, the optimal edge flow of the network is a $2$-PNE due to the presence of `positive' paths of length $2$ and $1$. However, considering path flows there exists a $1$-EF flow that induces  the optimal edge flow.  Also, the example has a path flow that is a $1.5$-UNE but it does not admit a path flow that is a $1$-UNE. More details  are discussed in Section~\ref{sec:satisfaction}, where these notions are formally introduced.

\input{related.tex}
\subsection{Contribution}
The recent influx of technology in traffic routing, the scale of traffic networks and globalization bring about a definite shift in the well studied routing games. The incomplete knowledge of users creates a dependence on routing technologies, giving more freedom to a central planner to mitigate the inefficiency originating from the selfish routing of users in the full information setting. In this work, we show that the path flows in the network may play a key role in achieving the full potential of such route planning mechanisms.  In particular, we clearly differentiate path flows from edge flows through the introduction of `positive' paths and `used' paths. Recall, a `positive' path is a path with all edges carrying nonzero flows under a given edge flow. Whereas a `used' path is a path with nonzero flow under a specific path flow. From the inherent differences of `positive' and `used' paths,  two new concepts, used Nash equilibrium (UNE) and envy free (EF) flow, naturally emerge as generalizations of the Wardrop equilibrium. We call the classical Wardrop equilibrium positive Nash equilibrium (PNE) because it essentially deals with `positive' paths. To the best of our knowledge, this distinction between positive and used paths has not been made explicit despite the rich literature developed on this topic for over half a decade. We also define the respective approximate versions of all the three solution concepts, i.e. $\theta$-PNE, $\theta$-UNE and $\theta$-EF for $\theta> 1$, where the distinction plays a critical role.

With the introduction of these three related concepts and their approximate versions, the first step in understanding them is to compare the flows against each other. We show that the $1$-UNE and the $1$-PNE are indeed identical and this helps in understanding why the `used' and the `positive' paths have not been explicitly differentiated before this work. But beyond this case the new concepts impose a hierarchical structure on the space of feasible flows. Specifically, we notice that $\theta$-PNE, $\theta$-UNE and $\theta$-EF flows are progressively larger sets, each containing the previous one, with promise of better tradeoff between the social welfare and fairness. In order to grasp the large separation between these concepts note that for some networks the $\theta$-UNE  is not contained in $\Omega(n\theta)$-PNE, where $n$ is the number of nodes in the network (Lemma~\ref{lemm:UNEvPNE} in Section~\ref{sec:hier}). 

Motivated from the classical study of the price of anarchy (PoA) of equilibrium flows we investigate the PoA of $\theta$-UNE and $\theta$-EF. In general we expect that as we move from the $\theta$-PNE to $\theta$-EF flows from a worst case perspective we will encounter flows with larger social cost. As a worst case example we show that the PoA can be unbounded for $1$-EF flows. However, we see that under the well used framework of variational inequality based PoA upper bounds~\cite{roughgarden2004bounding} both $\theta$-PNE and $\theta$-UNE admit the same bound on the PoA (Lemma~\ref{lemm:UNEvVI} in Section~\ref{sec:social}). Through a similar reasoning we show that the price of stability is non increasing from $\theta$-PNE to $\theta$-EF flows. 


Focusing on cost-efficient and fair flow design, the question of computing a $\theta$-PNE, a $\theta$-UNE or a $\theta$-EF flow with low social cost becomes one of the fundamental questions. We experience a temporary setback as the traditional convex optimization framework for computing the equilibrium and socially optimal flows fails here due to the non-convexity of the sets of $\theta$-PNE, $\theta$-UNE and $\theta$-EF flows for $\theta>1$. Formally, we prove (Theorem~\ref{thm:main_hardness} in Section~\ref{sec:complex}) that obtaining the best $\theta$-UNE or the best $\theta$-EF flow is NP-hard. Indeed given a socially optimal flow it is NP-hard to decide whether it admits a path flow decomposition which is $\theta$-UNE ($\theta$-EF). In a positive direction we show (Lemma~\ref{lemma:3Easy} in Section~\ref{sec:complex}) that for any `acylic' flow we can decide whether it is a $\theta$-PNE or not. As any `cyclic' flow can be made `acyclic' without increasing its social cost, the above result is sufficient for our design goal, i.e. balance social cost and fairness. However, we leave open the question of finding the best $\theta$-PNE flow ($\theta>1$). 

We further discuss how, at a conceptual level, the new ideas could be integrated with routing technologies (in Section~\ref{sec:design}). Drawing elements from different but related areas, we observe that minimization of modified latency functions can facilitate the calculation of a $\theta$-PNE flow with social cost guarantees. In particular, we use two techniques for bounding the social cost: 1) modified potential functions and 2) bounded tolls.  As a side note, following the ideas presented by Christodoulou et al.~\cite{christodoulou2011performance}, we explicitly articulate a technique to upper bound the price of stability for general functions and use it to extend the analysis of PoS for M/M/1 delay functions (Lemma~\ref{lemm:MM1} in Section~\ref{sec:design}).

In another direction, we deviate from the norm of deterministic flow design, and formalize the concept of randomization in flow design. We present (Theorem~\ref{thm:RR} in Section~\ref{sec:design}) an expression for the mean and a bound for the standard deviation of a path `used' by a typical user under this strategy.  The newly introduced concepts of $\theta$-UNE and $\theta$-EF flows play a crucial role in the variance reduction of this strategy.  The introduction of randomized routing in flow design may be of independent interest and we believe it can play an important role in emerging routing technologies.

%% file: related.tex
\subsection{Related Work}
The natural question of balancing the social welfare and user satisfaction is essential to practical traffic routing. Starting from the seminal work of Koutsoupias and Papadmitriou~\cite{koutsoupias1999worst}, quantifying the worst case inefficiency of various non-cooperative games, including routing games, quickly became an intense area of research. In a routing game with arbitrary latency functions the ratio between the cost of a Nash equilibrium (NE) flow to the cost of a socially optimal (SO) flow may grow unbounded, as shown by Roughgarden et al.~\cite{roughgarden2002bad}. A series of papers have focused on developing techniques for bounding the inefficiency of the NE flow~(e.g., \cite{correa2008geometric, roughgarden2002bad, harks2007price}). 
The next natural generalization led us to consider approximate NE flows with the hope that there exists some such flow which may improve the social welfare. The theoretical analysis by Caragiannis et al.~\cite{caragiannis2006tight} for linear latency functions and later by  Christodoulou et al.~\cite{christodoulou2011performance} for polynomial latency functions, corroborated this intuition. 
  
In a related thread of research, Jahn et al.~\cite{jahn2005system} formalized the notion of constrained system optimal, where additional constraints were added along with the flow feasibility constraints. The additional constraints were introduced to reduce the unfairness of the resulting flow.  Further, useful insights were obtained by Schulz et al.~\cite{schulz2006efficiency} about the social welfare and fairness of these constrained system optimal flows. Recently, there have been efforts~\cite{bertsimas2011price, bertsimas2012efficiency} in quantifying the inefficiency needed to guarantee fairness among users. The authors here define the `price of fairness' as the proportional decrease of utility under fair resource allocation. As mentioned earlier, in routing games the fairness of socially optimal flows under different but related definitions has been studied by Roughgarden~\cite{roughgarden2002unfair} and Correa et al.~\cite{correa2004computational}. Further, Correa et al.~\cite{correa2007fast} consider the fairness and  efficiency of min-max flows, where the objective is to minimize the maximum length of any used path in the network.  They note how different path flows affect the fairness in the network even when the induced edge flows are identical. 
 
In mechanism design with fully informed users various approaches for attaining better social welfare have been proposed and studied extensively. In a seminal work Beckman et al.~\cite{beckmann1956studies} showed that using marginal tolls one can induce SO as NE under tolled cost functions. Since then the idea of toll placement has been further generalized and studied under various practical settings, e.g. bounded tolls~\cite{hoefer2008taxing, bonifaci2011efficiency, jelinek2014computing} and heterogeneous users~\cite{cole2003pricing, fleischer2004tolls}. A Stackelberg equilibrium, where a fraction of users is willing to cooperate and follow the routes dictated by the central planner, and its variations~\cite{karakostas2009stackelberg, swamy2012effectiveness} have also been studied as an alternative. However, as new technologies play a crucial role in shifting in user behavior, various incomplete information models have been introduced.  Acemoglu et al.~\cite{acemoglu2016informational} has discussed an informational Nash equilibrium where users converge to an equilibrium with partial knowledge of the network structure. In a related setup where each user's information is limited to a common prior on the latency functions, Vasserman et al.~\cite{vasserman2015implementing} considered a mediated Bayesian Nash equilibrium (BNE). Under an incentive compatible mediation strategy they studied the cost of the BNE in a parallel arc network and showed it is bounded by the number of edges. Other work has also studied mediated games with tolls~\cite{rogers2015inducing} and without tolls~\cite{kearns2014mechanism} where the focus has been truthful mechanism design using differential privacy techniques.  

%% file: prelimv1.tex
\section{Preliminaries}\label{sec:prelim}
\subsection{Network and Flows}
\textbf{Network.} We are given a directed graph $G(V,E)$ with vertex set $V$, edge set $E$, and a set of commodities 
$\mc{K}=\{1,2,\dots,K\}$.  Each commodity $k \in \mc{K}$ is associated with a source $s_k$ and a sink $t_k$.  We denote $\mc{T} = \{(s_k,t_k)\}_{k \in \mc{K}}$ as the collection of the source-sink pairs for all commodities.
Also, for each commodity $k \in \mc{K}$, let $\mc{P}^k$ be the set of directed simple paths in $G$ from $s_k$ to $t_k$, and let $d_k > 0$ be the demand associated with commodity $k$. Define $\mc{P}:=\cup_{k \in \mc{K}}\mc{P}^k$ to be the set of paths over all commodities and $\bm{d}:=(d_k)_{k \in \mc{K}}$ to be the vector of the demands. Each edge $e \in E$ is given a load-dependent \emph{latency function} $\ell_e(x)$, assumed to be nonnegative, differentiable, and nondecreasing. Moreover, we assume $x \ell_e(x)$ is convex with respect to $x$.  We shall abbreviate an instance of the problem by the quadruple $\mc{G}=(G(V,E),\mc{T}, 
\{\ell_e\}_{e\in E}, \bm{d})$.


\smallskip\noindent\textbf{Flows.} Given an instance $\mc{G}$, the collective decisions of users in commodity $k \in \mc{K}$ can be encoded in two ways, as a path flow $\bm{f}^k=(f_{\pi}^k)_{\pi \in \mc{P}}$ and as an edge flow $\bm{x}^k = (x_e^k)_{e\in E}$. These two representations are related as $x_e^k = \sum_{\pi \in \mathcal{P}^k:\pi \owns e} f_{\pi}^k$. We can also consider the collective decisions of users of all commodities together by defining the {\pathdecomp} $\bm{f} = \sum_{k \in \mc{K}} \bm{f}^k$ and the edge flow $\bm{x} = \sum_{k \in \mc{K}} \bm{x}^k$.
There may exist multiple {\pathdecomps} corresponding to an edge flow $\bm{x}$ and we denote the set of such decompositions as $\mc{D}_p(\bm{x})$. Denote the feasible edge flows by $\mathcal{D}_E.$\footnote{
For node $u\in V$, $E_u^+$ denote the set of its outgoing edges and $E_u^-$ denote the set of its incoming edges.  $\mc{D}_E$ is the set of vectors that satisfies the flow conservation equations:
\Scale[0.7]{
\mathcal{D}_E = \left\{\bm{x}: x_e=\sum_{k \in \mc{K}} x_e^k,
\sum_{e\in E_{u}^+} x_e^k -\sum_{e\in E_{u}^-} x_e^k = d_k \left(\mathbbm{1}_u (s_k) - \mathbbm{1}_u (t_k)\right),\forall e\in E,\\
\forall u\in V, \forall k \in \mc{K}\right\}.
}
}
We can define the feasible region for all possible path flows as $\mc{D}_p = \cup_{\bm{x} \in \mc{D}_E} \mc{D}_p(\bm{x})$.
  
We further differentiate a \emph{positive} path from a \emph{used} path in the following definitions.
 
\begin{definition}[Positive path]
For an edge flow vector $\bm{x}$, we call a path $\pi\in \mathcal{P}$ \emph{positive} for commodity $k \in \mc{K}$ if for all edges $e \in \pi$, $x_{e}^k > 0$.  
For each commodity $k \in \mc{K}$, we can define the set of \emph{positive} paths under edge flow $\bm{x}$ as $\mathcal{P}_{+}^k(\bm{x}) = \left\{p: p \in \mathcal{P}^k, \forall e \in p, x_e^k>0  \right\}$.  Further, the set of all positive paths for all commodities under edge flow $\bm{x}$ can be defined as $\mc{P}_{+}(\bm{x}) = \cup_{k \in \mc{K}} \mc{P}_{+}^k(\bm{x})$.
\end{definition}
     
\begin{definition}[Used path]
For a path flow $\bm{f}$, we call a path $\pi\in \mathcal{P}$ \emph{used} by commodity $k \in \mc{K}$ if $f_{\pi}^k > 0$ and \emph{unused} otherwise. For each commodity $k \in \mc{K}$, we can define the set of \emph{used} paths under path flow decomposition $\bm{f}$ as $\mathcal{P}_u^k(\bm{f}) = \{p: p\in \mathcal{P}, f_p^k >0\}$.  Further, the set of all used paths for all commodities under path flow decomposition $\bm{f}$ can be defined as $\mc{P}_u(\bm{f}) = \cup_{k \in \mc{K}} \mc{P}_u^k(\bm{f})$.
\end{definition}

\begin{remark}
Note that a used path is always positive but a positive path may be unused depending on the particular path flow decomposition.
\end{remark}

\subsection{Costs and Equilibria}
\textbf{Costs.} Under a path flow $\bm{f} \in \mc{D}_p$, the cost (latency) of a path $\pi$ is defined to be the sum of latencies of edges along the path: 
$\ell_{\pi}(\bm{f}) = \ell_{\pi}(\bm{x})=\sum_{e \in \pi}\ell_e(x_e)$ for $\bm{f} \in \mc{D}_p(\bm{x})$.

\begin{definition}[Social cost and socially optimal flow]
The \emph{social cost} (SC) of a flow $\bm{x} \in \mc{D}_E$ is the total latency in the network under the flow, $SC(\bm{x})=\sum_{e\in E}x_e \ell_{e}(x_e)$. The social cost of a path flow $\bm{f} \in \mc{D}_p$ is $SC(\bm{f})=SC(\bm{x_f})$, where $\bm{x_f}$ is the edge flow induced by $\bm{f}$.  We sometimes refer to the social cost simply as {\em cost}.
A flow with minimum social cost among all feasible flows is called a \emph{socially optimal} flow or simply, a \emph{social optimum}.  The set of socially optimal edge flows is denoted by 
$$SO_E = \{\bm{x} \in \arg\min SC(\bm{x}) \}.$$
Also, we denote the set of socially optimal path flows by
$$ SO_p= \{\bm{f} \in \arg\min SC(\bm{f}) \}.$$
\end{definition}

\smallskip\noindent 
\textbf{Equilibrium.} We assume that users are {\em nonatomic}, namely there are infinitely many users that are infinitesimally small.  As such, a single user controls an infinitesimally small fraction of flow and her routing choice does not unilaterally affect the costs experienced by other users.  This fact is captured by the definition of equilibrium below.

\begin{definition} (Nash Equilibrium)\footnote{The Nash equilibrium in nonatomic routing games is also commonly known as Wardrop equilibrium. }
A path flow $\bm{f}$ is a  \emph{Nash Equilibrium} if for any commodity $k\in \mathcal{K}$ and any used path $p\in \mathcal{P}_{u}^k(\bm{f})$ we have $\ell_p(\bm{f}) \leq  \ell_q(\bm{f})$, for all paths $q\in \mathcal{P}^k$. 
\end{definition}

Given a Nash equilibrium, we can measure its quality by comparing its cost with the cost of the socially optimal flow.  This idea is often formalized as the \emph{price of anarchy} and the \emph{price of stability} which we define below. Since our scope is to examine user oriented solution concepts other than  the Nash Equilibrium, we generalize the classic definitions of the price of anarchy and stability to apply to an arbitrary set of flows $\mc{F}$. If $\mc{F}$ is the set of Nash equilibria, we get the standard definition for the price of anarchy and price of stability.

\begin{definition}[Price of Anarchy and Price of Stability]  Given an instance $\mc{G}$
and a set of (feasible) flows $\mc{F}$, we define the price of anarchy (PoA) as the ratio of the maximum social cost of any flow in $\mc{F}$  to the socially optimal cost. The price of stability (PoS) is the ratio of the minimum social cost of any  flow in $\mc{F}$ to the socially optimal cost.  The PoA and PoS are formally expressed as: 
\begin{flalign}\label{eq:PoA}
PoA(\mc{F}) = \max \left\{  \ \frac{SC(\bm{f})}{SC(\bm{x}^*)}: \bm{f} \in \mc{F}, \bm{x}^* \in SO_E \right\}.
\end{flalign}
\begin{flalign}\label{eq:PoS}
PoS(\mc{F}) = \min \left\{  \ \frac{SC(\bm{f})}{SC(\bm{x}^*)}: \bm{f} \in \mc{F}, \bm{x}^* \in SO_E \right\}.
\end{flalign}
We may define the PoA and the PoS over sets of instances. For a set of instances, its PoA and PoS equals the maximum  PoA and PoS among the instances in the set, respectively. We will use this definition when examining  instances with latency functions in class $\mc{L}$ (for some $\mc{L}$), and it will be clear from the context.
\end{definition}

The PoA and the PoS for the set of Nash equilibria coincide in nonatomic selfish routing instances, since there is only one Nash equilibrium (up to edge costs). In contrast, for the sets that we consider in this work and introduce in the following section (e.g., the set of approximate Nash equilibria)  the PoA and the PoS may get different values.

%% file: solutionconcepts.tex
\section{Solution Concepts}\label{sec:satisfaction} 
Here we give the formal definition of the solution concepts we introduced in Section~\ref{sec:intro}.  We also provide an example to illustrate their differences, and prove that each solution concept may correspond to a non-convex set of flows. 

\begin{definition}[$\theta$-PNE]\label{def:ApproxPNE}
Given a network $\mc{G}$, an edge flow $\bm{x}$ is a $\theta$-Positive Nash Equilibrium ($\theta$-PNE) flow if for any commodity $k\in \mathcal{K}$ and any positive path $p\in \mathcal{P}_{+}^k(\bm{x})$ we have $\ell_p(\bm{x}) \leq \theta \ell_q(\bm{x})$, for all paths $q\in \mathcal{P}^k$. 
We may call a path flow $\bm{f}$  a $\theta$-Positive Nash Equilibrium, if $\bm{f}\in \mc{D}_p(\bm{x})$, for some $\theta$-Positive Nash Equilibrium edge flow $\bm{x}$.
\end{definition}

\begin{definition}[$\theta$-UNE]\label{def:ApproxUNE}
Given a network $\mc{G}$, a {\pathdecomp} $\bm{f}$ is a $\theta$-Used Nash Equilibrium ($\theta$-UNE) flow if for any commodity $k\in \mathcal{K}$ and any used path $p\in \mathcal{P}_{u}^k(\bm{f})$ we have $\ell_p(\bm{f}) \leq \theta \ell_q(\bm{f})$, for all paths $q\in \mathcal{P}^k$. 
\end{definition}

The definition of $\theta$-UNE  corresponds to that of $\theta$-approximate Nash equilibrium used thus far in the literature.
For $\theta=1$, $1$-UNE and $1$-PNE (or simply PNE) coincide, as we show in Lemma~\ref{lemm:UNEvPNE}, and they correspond to the Nash equilibrium, which has been studied extensively. 
It turns out that every PNE of an instance solves the convex optimization problem  
$\{ \sum_{e\in E}\int_{0}^{x_e} \ell_e(x)dx : \bm{x}\in \mc{D}_E\}$,
which as a consequence yields the uniqueness of PNE up to edge costs, i.e. for all $e$ and any two PNE flows, $\bm{x}$, $\bm{x}'$, $\ell_e(x_e)=\ell_e(x'_e)$. This implies for any commodity $k$, all the positive paths have length $L^k_{NE}$ which is called the Nash length of that commodity.   

\begin{definition}[$\theta$-EF]\label{def:ApproxEF}
Given a network $\mc{G}$, a {\pathdecomp} $\bm{f}$ is $\theta$-Envy Free if for any commodity $k\in \mathcal{K}$ and any used path $p\in \mathcal{P}_u^k(\bm{f})$ we have $\ell_p(\bm{f}) \leq \theta \ell_q(\bm{f})$, for all used paths $q\in \mathcal{P}_u^k(\bm{f})$. 
\end{definition} 

For an instance, we may use $\theta$-PNE, $\theta$-UNE or $\theta$-EF to describe the set of $\theta$-PNE, $\theta$-UNE or $\theta$-EF flows respectively, which will be clear from the context, and we may omit $\theta$ to represent $\theta=1$. Also, by \emph{incentive conditions} we refer to the conditions (inequalities) used to define these flows. Additionally, we may refer to all $\theta$-PNE, $\theta$-UNE and $\theta$-EF flows as $\theta$ fair flows.  The reason for that comes from the fact that their incentive  conditions have inherent the comparison of the maximum used path cost with the minimum (used) path cost, which in some sense describes how (un)fair the flow for players on the maximum cost paths compared to the cost of the lowest cost (used) paths is. Similar notions of (un)fairness have been examined in the past, e.g., by Rougharden \cite{roughgarden2002unfair} and Correa et al. \cite{correa2007fast}. 

\begin{figure}
	\centering
	\begin{subfigure}[b]{0.45\linewidth}
		\centering
		\includegraphics[width=0.9\linewidth]{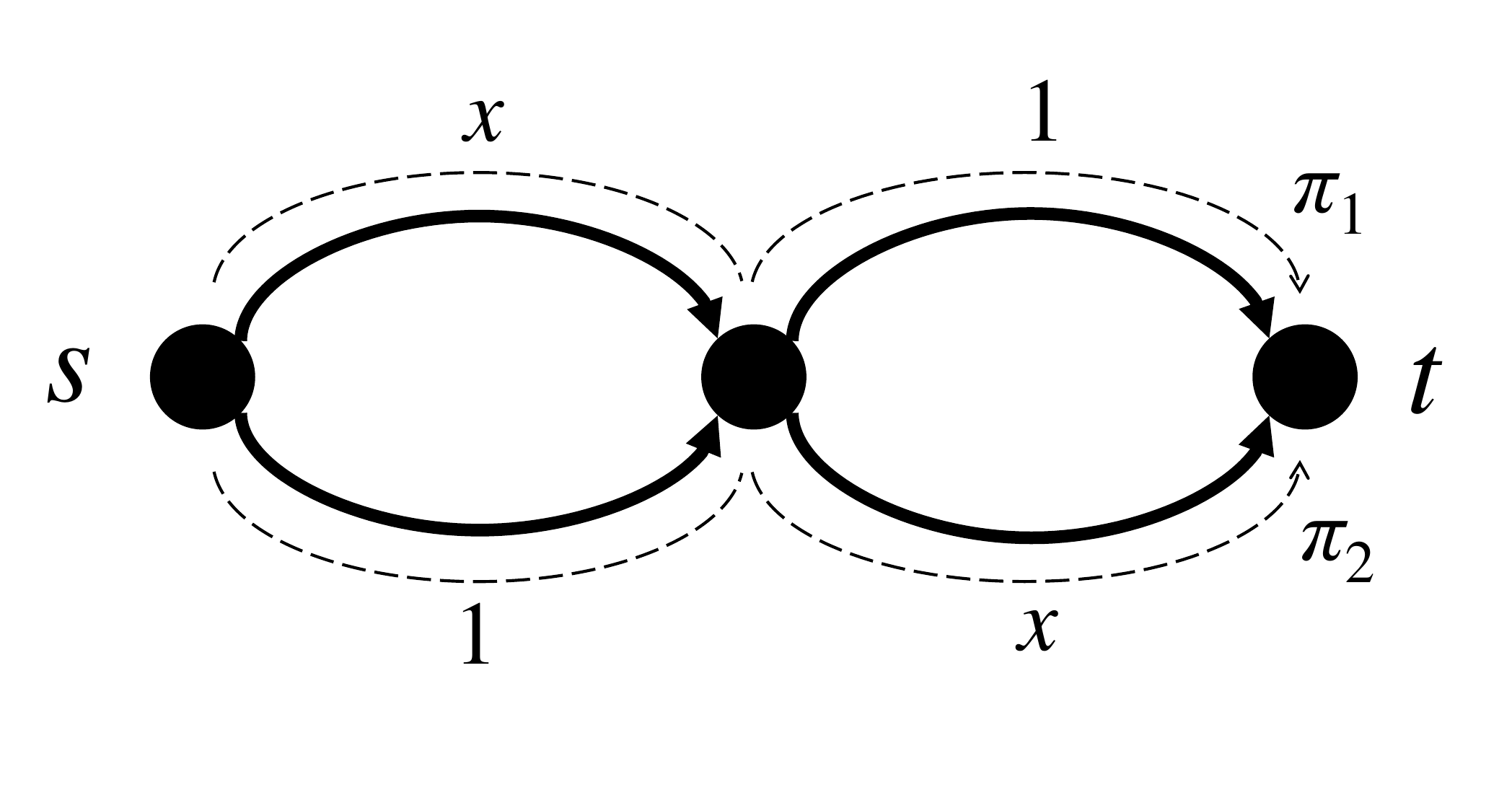}
		\caption{Paths $\pi_1$ and $\pi_2$ have $1/2$ unit of flow.  This path flow assignment is a social optimum.}
		\label{fig:ex_flows1}
	\end{subfigure}\hspace{0.05\linewidth}
	\begin{subfigure}[b]{0.45\linewidth}
		\centering
		\includegraphics[width=0.9\linewidth]{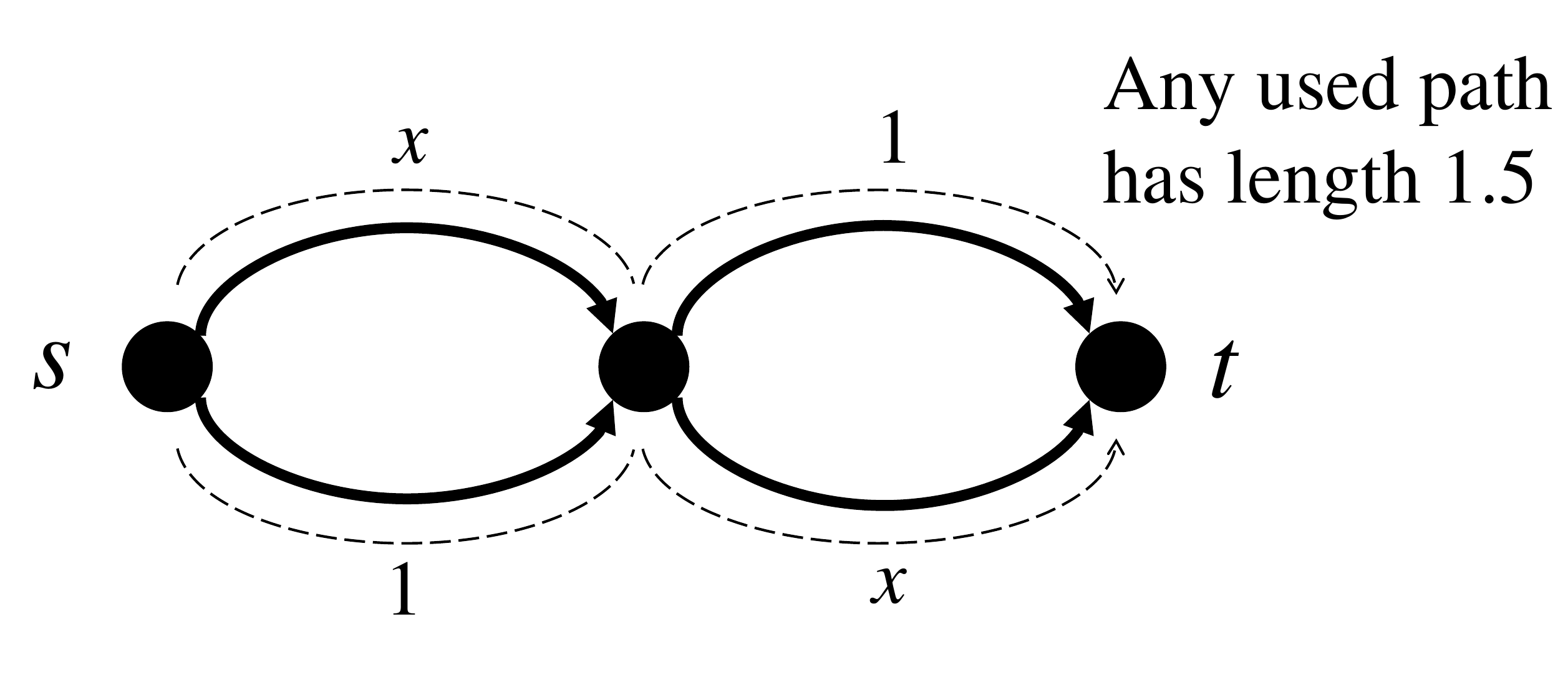}
		\caption{The path flow assignment in Figure~\ref{fig:ex_flows1} is $1$-EF but not $1$-UNE.}
		\label{fig:ex_flows01}
	\end{subfigure}
	
	\begin{subfigure}[b]{0.45\linewidth}
		\centering
		\includegraphics[width=0.9\linewidth]{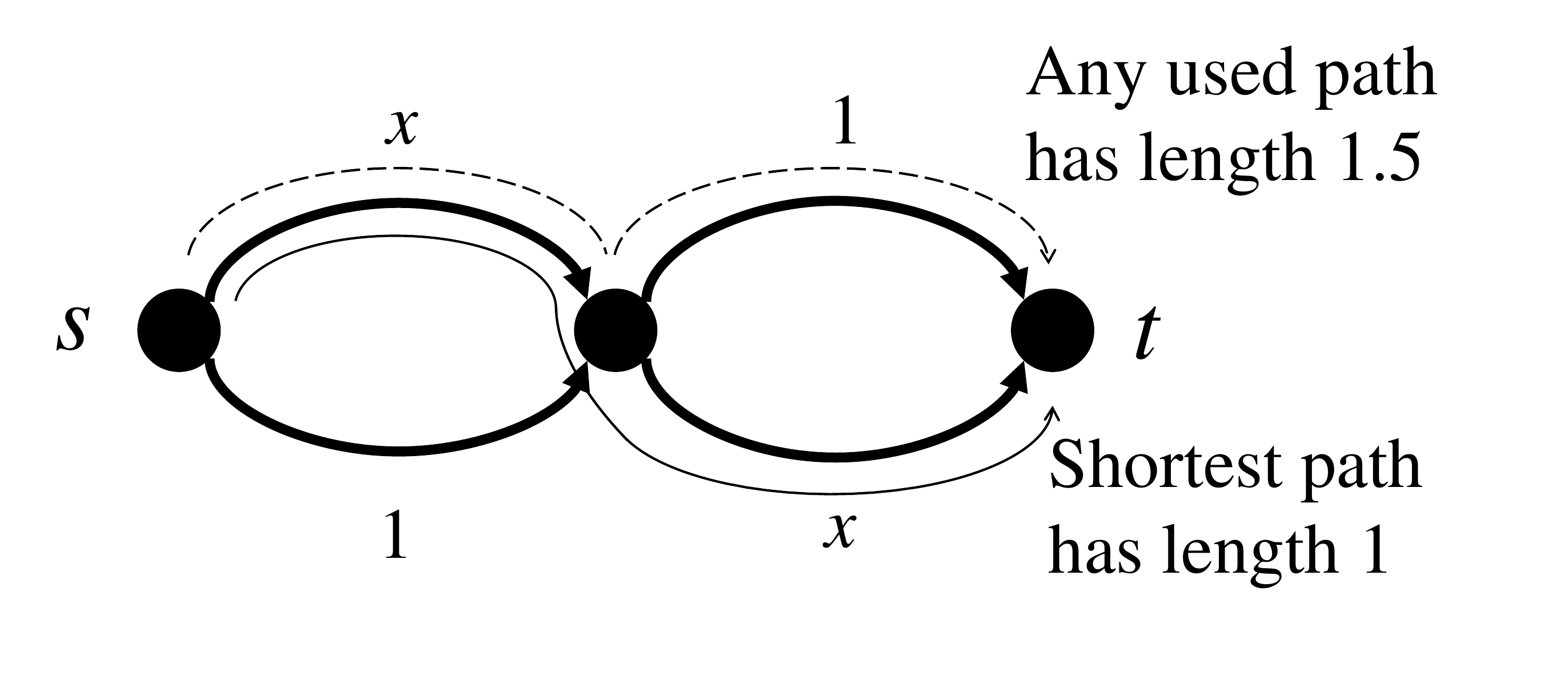}
		\caption{The path flow assignment in Figure~\ref{fig:ex_flows1} is $1.5$-UNE but not $1.5$-PNE.}
		\label{fig:ex_flows2}
	\end{subfigure}\hspace{0.05\linewidth}
	\begin{subfigure}[b]{0.45\linewidth}
		\centering
		\includegraphics[width=0.9\linewidth]{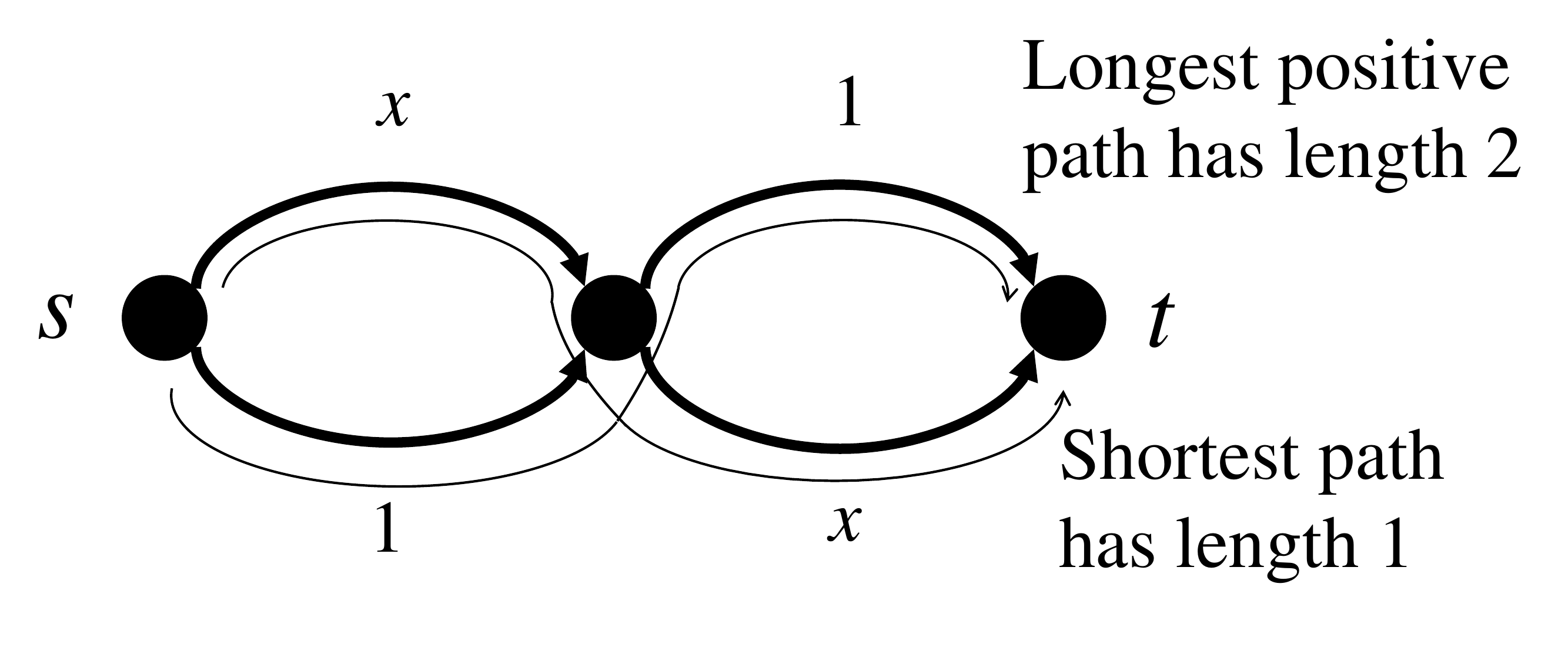}
		\caption{The path flow assignment in Figure~\ref{fig:ex_flows1} is $2$-PNE.}
		\label{fig:ex_flows3}
	\end{subfigure}
	\caption{Example illustrating the three solution concepts $\theta$-UNE, $\theta$-PNE and $\theta$-EF. }
	\label{fig:ex_flows}
\end{figure}

To see how these concepts differ from each other, we give an example in Figure~\ref{fig:ex_flows}.  Suppose the instance is given as shown in Figure~\ref{fig:ex_flows1} and there is a single commodity that routes a unit demand  from $s$ to $t$.  We consider the path flow that routes $1/2$ of the demand through path $\pi_1$ and routes $1/2$ through path $\pi_2$.  It is easy to verify that this is indeed a socially optimal flow.  Next, let us find the appropriate sets that this path flow assignment belongs to with respect to the solution concepts we described above:
\begin{enumerate}
	\item \textbf{$1$-EF:} As shown in Figure~\ref{fig:ex_flows01}, it is easy to verify that this flow is a $1$-EF as each used path has the same length.
	\item \textbf{$1.5$-UNE:} As shown in Figure~\ref{fig:ex_flows2}, any used path has length $1.5$ and the shortest path in the graph has length $1$.  Hence the flow is a $1.5$-UNE as the length of each used path is within a factor $1.5$  of any path.
	\item \textbf{$2$-PNE:} As shown in Figure~\ref{fig:ex_flows3}, for a PNE, we have to take all positive path\evn{s} into account.  Since the longest positive path has length $2$, this flow assignment is not a $1.5$-PNE flow.  Instead, we can see that any positive path is within a factor $2$ of any path.  Hence, this flow is a $2$-PNE.
\end{enumerate}

Our goal is to examine the properties of $\theta$ fair flows and provide ways to obtain such flows with good social cost. Regarding the second direction,  in general, the sets of $\theta$-PNE, $\theta$-UNE, and $\theta$-EF flows may not be convex and may contain multiple path flows, which raises the level of difficulty for computing good or optimal such flows. Next we present an example that demonstrates the non-convexity of these sets.

\begin{proposition} There exists a network $\mc{G}$ and $\theta>1$ such that the sets $\theta$-PNE, $\theta$-UNE, and $\theta$-EF are not convex.
\end{proposition}
\begin{proof}
The instance in Figure~\ref{fig:ex_noncvx2} demonstrates that the sets $\theta$-PNE, $\theta$-UNE, and $\theta$-EF are all non convex.  Consider a commodity routing from $s$ to $t$ with unit demand. 
Then, consider the following two flow assignments.
The first one routes all demand along the path $s-u-v-t$.  In this case, path $s-u-v-t$ is the only positive path and has cost equal to $2$.  It is easy to verify that this is a $3/2$-PNE, $3/2$-UNE, and $3/2$-EF.  The second flow routes $2/3$ of the demand along  path $s-u-t$ and routes $1/3$ along path $s-v-t$.  Path $s-u-t$ has cost equal to $1$ and path $s-v-t$ has cost equal to $3/2$.  It is easy to verify that this is a $3/2$-PNE, $3/2$-UNE, and $3/2$-EF as well.  However, if we take the convex combination of these two assignments evenly, then we can find that the path $s-v-t$ has cost equal to $11/6$ and the path $s-u-v-t$ has cost equal to $7/6$, and thus their ratio is greater than $3/2$.  This shows that this combined flow is neither a $3/2$-PNE, a $3/2$-UNE, nor a $3/2$-EF, and hence they are not convex sets.
\end{proof}

\begin{figure}
	\centering
	\includegraphics[width=0.3\linewidth]{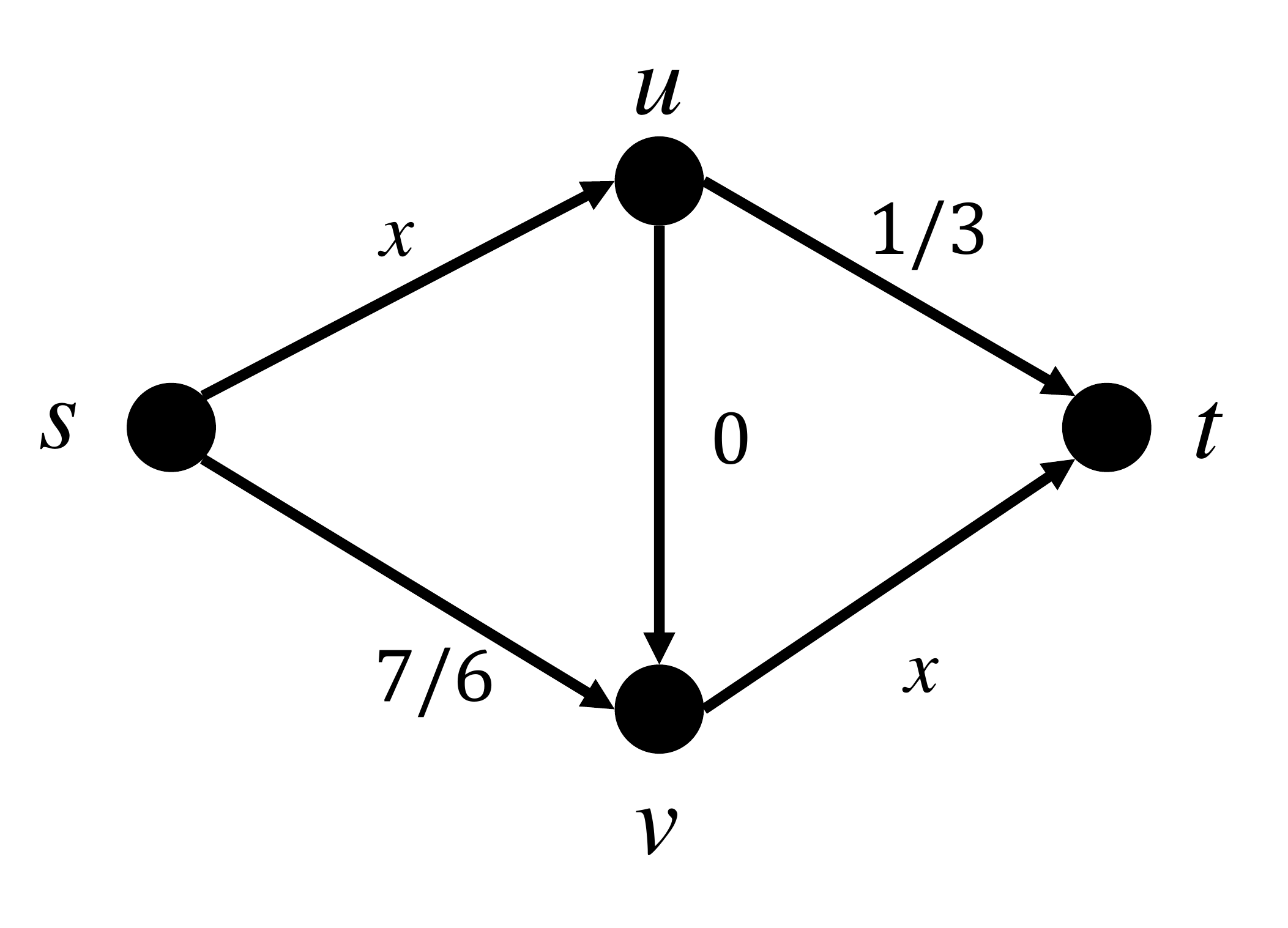}
	\caption{Non-convexity of $\theta$-flows}
	\label{fig:ex_noncvx2}
\end{figure}


%% file: hier.tex
\section{Solution Concepts Hierarchy}\label{sec:hier}
In this section, we discuss the interrelation between the solution concepts defined in Section~\ref{sec:satisfaction}. 
For completeness, we first state the trivial relation between different $\theta$-flows of the same type.
\begin{proposition}\label{lemm:intra}
For $\theta' > \theta\geq 1$ and $\mc{F} \in \{\text{PNE, UNE, EF}\}$, $\theta$-$\mc{F}$ $\subseteq$ $\theta'$-$\mc{F}.$
\end{proposition}

In the rest of the section, we discuss the relations between different types of flows. 
\begin{lemma}\label{lemm:PNEvUNEvEF}
For any $\theta\geq 1$, we have the containment 
$$\theta\text{-PNE} \subseteq \theta\text{-UNE} \subseteq \theta\text{-EF}.$$
\end{lemma}
\begin{proof}
The first containment is due to the fact that any used path in a network $\mc{G}$ is a positive path in $\mc{G}$. Let $\bm{f}$ be a $\theta$-PNE, then for any commodity $k$, $$\max_{p\in \mc{P}^k_u} \ell_p(\bm{f}) \leq \max_{p\in \mc{P}^k_{+}} \ell_p(\bm{f}) \leq \theta\min_{p\in \mc{P}^k} \ell_p(\bm{f}).$$ Therefore, $\bm{f}$ is a $\theta$-UNE. 

Next, for the second containment let $\bm{f}$ be a $\theta$-UNE. We have for any commodity $k$, $$\max_{p\in \mc{P}^k_u} \ell_p(\bm{f}) \leq \theta\min_{p\in \mc{P}^k} \ell_p(\bm{f}) \leq \theta\min_{p\in \mc{P}^k_{u}} \ell_p(\bm{f}).$$ We conclude that $\bm{f}$ is a $\theta$-EF.
\end{proof}

\begin{proposition}\label{lemm:EFvUNE}
For any $\theta \geq 1$ there exists a network $\mc{G}$ s.t. $1\text{-EF} \not\subset \theta\text{-UNE}.$
\end{proposition}
\begin{proof}
Consider the two parallel link networks with constant latency $1$ in the lower link and latency $\ell(x)=x$ for the upper link. There is a commodity with demand one between the two nodes. Consider the flow using only the lower link. This is a $1\text{-EF}$ flow but as the minimum path has length equal to $0$ it can not be classified as $\theta\text{-UNE}$ for any $\theta$.
\end{proof}

Next, a more detailed relation between $\theta$-UNE and $\theta$-PNE is shown in Lemma~\ref{lemm:UNEvPNE}. 
\begin{lemma}\label{lemm:UNEvPNE}
	The following statements are true:
	\begin{enumerate}
		\item For any $\theta > 1$ and multi-commodity network $\mc{G}$ with $n$ nodes, $\theta\text{-UNE} \subset ((n-1)\theta)\text{-PNE}.$
		\item For any $\theta\geq 1.5$ there exists a single commodity network $\mc{G}$ with $n$ nodes such that $\theta\text{-UNE} \not\subset ((n-3)\theta/3)\text{-PNE}.$ 
		\item (Equivalence of $1$-PNE and $1$-UNE)  For any path flow $\bm{f}$, let $\bm{x}$ be the edge flow induced by $\bm{f}$.  Then $\bm{f} \in$ $1$-UNE if and only if $\bm{x} \in$ $1$-PNE.
	\end{enumerate}
\end{lemma}
\begin{proof}
First, consider a multi-commodity network with $n$ nodes and let  $k$ be any of its commodities. Let $\bm{f}$ be a $\theta\text{-UNE}$ with edge flow $\bm{x}$. We have 
\begin{align*}
\frac{\max_{p\in \mc{P}^k_{+}}\ell_p(\bm{f})}{\min_{p\in \mc{P}^k}\ell_p(\bm{f})} 
&\leq (n-1)\frac{\max_{e: x_e^k >0} \ell_e(\bm{x})}{\min_{p\in \mc{P}^k} \ell_p(\bm{f})} \\
&\leq (n-1)\frac{\max_{p\in \mc{P}^k_{u}}\ell_p(\bm{f})}{\min_{p\in \mc{P}^k}\ell_p(\bm{f})} \leq (n-1)\theta.
\end{align*} 

\begin{figure}
\centering
\includegraphics[width=0.57\linewidth]{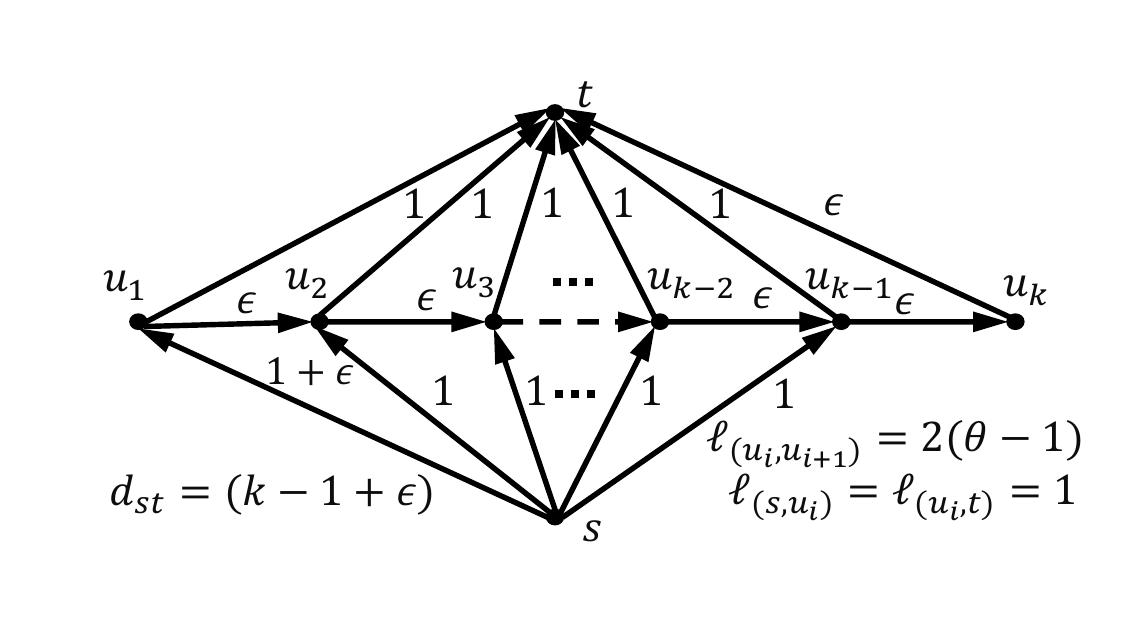}
\caption{$\theta$-UNE vs $\theta$-PNE in Lemma \ref{lemm:UNEvPNE}.}
\label{fig:UNEvPNE-SC}
\end{figure}

For the second part, consider the network in Fig.~\ref{fig:UNEvPNE-SC} with $n=k+2$ nodes. The edge flows are: $(s,u_1) = 1+\epsilon$, $(s,u_i)=1$ for $i=\{2,\dots, k-1\}$, $(u_k,t) = \epsilon$, $(u_i,t) = 1$ for  $i=\{1,\dots, k-1\}$ and $(u_i,u_{(i+1)}) = \epsilon$ for  $i=\{1,\dots, k-1\}$. Let, for  $i=\{1,\dots, k-1\}$, the edges $(u_i,u_{(i+1)})$ have latency $2(\theta - 1)$ and the remaining edges have latency $1$ under this flow.  The path decomposition $f(s,u_i,t)= 1$, $f(s,,u_i, u_{(i+1)},t)= \epsilon$ for  $i=\{1,\dots, k-1\}$ is a $\theta\text{-UNE}$ for this network. Whereas, the ratio of the minimum positive path,  $\ell((s,u_1,t))= 2$, and the maximum positive path $\ell((s,u_1,\dots, u_k, t))= 2+2(k-1)(\theta-1)$ is $\left(1+(k-1)(\theta -1)\right) > (n-3)\\theta/3$ for $\theta \geq 1.5$. The second part of the lemma follows.

For the third part, the \emph{if} direction follows from Lemma~\ref{lemm:PNEvUNEvEF}.  For the \emph{only if} direction, we first note that under UNE, every used path in the same commodity has the same length, which is the Nash length $L^k_{NE}$ (for commodity $k$).  It suffices to show that all positive paths in commodity $k$ have length $L^k_{NE}$ as well.  We prove this by contradiction.  Consider an instance $\mc{G}$ and a UNE path flow $\bm{f}$.  Assume there is a positive path $\pi$ in commodity $k$ with length not equal to $L_{NE}^k$. From the definition of UNE the length of path $\pi$ must be strictly greater than $L^k_{NE}$.  Let $e_1,\dots,e_r$ be the edges in $\pi$ in the order of traversal.  For each $e_i \in \pi$, we choose a path $\pi_i=\pi_i^s-e_i-\pi_i^t$ used by commodity $k$ that uses the edge $e_i$.  Note that under this notation $\pi_1^s=\emptyset$ and $\pi_r^t=\emptyset$. We can see that the sum of the lengths of these paths is $\sum_{i=1}^{r} \ell_{\pi_i}(\bm{f}) = rL^k_{NE}$. Now, consider the paths $\pi_1',\dots,\pi_{r-1}'$, where $\pi_i'=\pi_{i+1}^s-\pi_i^t$.  We can find that 
$$
rL^k_{NE}=\sum_{i=1}^{r} \ell_{\pi_i}(\bm{f}) = \sum_{i=1}^{r-1} \ell_{\pi_i'}(\bm{f}) + \sum_{i=1}^{r} \ell_{e_i}(\bm{f}) = \sum_{i=1}^{r-1} \ell_{\pi_i'}(\bm{f}) + \ell_{\pi}(\bm{f})
$$
According to the definition of UNE, we can see that for $i\in\{1,\dots,r-1\}$, $\ell_{\pi_i'}(\bm{f}) \ge L^k_{NE}$.  Therefore,  we have $\ell_{\pi}(\bm{f}) \le L^k_{NE}$, which contradicts the assumption. As a consequence, for each $k$, all positive paths of commodity $k$ have length equal to $L^k_{NE}$.

   
\end{proof}

%% file: cost.tex
\section{Quality of $\theta$-Flows}\label{sec:social}
In this section, we analyze the cost and fairness of the solution concepts introduced in Section~\ref{sec:satisfaction}.
As noted in Section~\ref{sec:satisfaction}, the $\theta$ flows are not unique for $\theta>1$ and that implies that potentially under each solution concept we can have a range of attainable costs. We present the upper bounds on the PoS and PoA, defined respectively in \eqref{eq:PoA} and \eqref{eq:PoS},  for the flows under the three solution concepts. We compare the social cost of the $\theta$ flows with the socially optimal flow.
 
\textbf{Price of Anarchy.} Starting from Correa et al.~\cite{correa2008geometric},  there has been a unifying approach of bounding the PoA using the variational inequality formulation of the Nash equilibrium flow.  Christodoulou et al.~\cite{christodoulou2011performance}  extended the idea of using a variational inequality to formulate an approximate equilibrium, specifically a $\theta$-PNE, and to give new bounds for the price of anarchy of a $\theta$-PNE. We first show that the $\theta$ variational inequality encompasses both the $\theta$-PNE and $\theta$-UNE. Therefore, we can use the well established technique to give an upper bound for both types of flows.

\begin{lemma}\label{lemm:UNEvVI}
If a flow $\bm{f}$ is a $\theta$-UNE with edge flow $\bm{x}$, then it satisfies the following variational inequality for $\theta\geq 1$,
\begin{align}\label{eq:thetaVI}
\sum_e x_e\ell_e(x_e) \leq \theta\sum_e x'_e\ell_e(x_e), ~\forall \bm{x}'\in \mc{D}_E. 
\end{align}                                                                                                                           
Further, there exists a single commodity network and a flow $ \bm{f}'' \in \mc{D}_p(\bm{x}'')$ that satisfies the above inequality but is not a $\theta$-UNE.
\end{lemma}
\begin{proof}
The proof of the first part follows closely the proof of Theorem $1$ in Christodoulou et al~\cite{christodoulou2011performance}. Let $\bm{f} \in \mc{D}_p(\bm{x})$ be a $\theta$-UNE and  $\bm{f}' \in \mc{D}_p(\bm{x}')$ be any other feasible flow in the network. From the definition of $\theta$-UNE, for any commodity $k$, for any used path $p\in \mc{P}^k_u$ and for any other path $p'\in \mc{P}^k$ we have $\sum_{e\in p}\ell_e(x_e) \leq \theta \sum_{e\in p'}\ell_e(x_e).$  Further,  taking the summation of the flow weighted path latency over all pairs of paths in commodity $k$, we obtain the following:
\begin{align*}
\sum_{\substack{p\in \mc{P}^k_u(\bm{f})\\ p'\in \mc{P}^k(\bm{f}')}} f_p^k f_p^{k\prime}\sum_{e\in p}\ell_e(x_e) 
&\leq \theta\sum_{\substack{p\in \mc{P}^k_u(\bm{f})\\ p'\in \mc{P}^k(\bm{f}')}} f_p^k f_p^{k\prime}  \sum_{e\in p'}\ell_e(x_e), \\
\sum_{p'\in \mc{P}^k(\bm{f}')} f_p^{k\prime} \sum_{e\in E}x_e^k\ell_e(x_e) 
&\leq \theta \sum_{p\in \mc{P}^k_u(\bm{f})} f_p^k  \sum_{e\in E}x_e^{k\prime}\ell_e(x_e), \\
\sum_{e\in E}x_e^k\ell_e(x_e)&\leq \theta \sum_{e\in E}x_e^{k\prime}\ell_e(x_e).
\end{align*} 
The last inequality follows due to $\sum_{p'\in \mc{P}^k(\bm{f}')} f_p^{k\prime} = \sum_{p\in \mc{P}^k_u(\bm{f})} f_p^k = d_k>0$. Finally, taking summation over all commodities gives $\sum_e x_e\ell_e(x_e) \leq \theta\sum_e x'_e\ell_e(x_e)$.

Consider Pigou's network with demand $1$, top edge with latency $\ell_t(x) = \epsilon x$ and bottom edge with latency $\ell_b(x) = L$. For a given $\theta>1$, choose $\delta>0$ small (to be specified later).
 Consider the flow $\bm{f}''$ equal to
 $(1-\delta\epsilon/L)$ in the top link and $\delta\epsilon/L$ in the bottom link. 
 The social cost of this flow is $\left(\delta\epsilon+\epsilon(1-\delta\epsilon/L)^2\right)$. The feasible flow minimizing the right hand side of Inequality \eqref{eq:thetaVI} is flow of $1$ through top link for $L>\epsilon$. Further, for any $\theta$, there is a $\delta$ small enough such that Inequality  \eqref{eq:thetaVI} holds, since, by the above, it suffices to have $\delta\epsilon+ \epsilon(1-\delta\epsilon/L)^2 \leq\theta \epsilon(1-\delta\epsilon/L)
\Leftrightarrow \theta\geq 1+ \frac{L^2\delta-L\delta \epsilon+\delta^2 \epsilon^2}{L(L-\delta\epsilon)} $. 
However, $\bm{f}''$ is not $(L/\epsilon)$-UNE and this approximation factor can be made arbitrarily large making $\epsilon$ small enough. 
\end{proof}

\begin{remark}
Here we have shown that the variational inequality is a sufficient condition for a flow to be $\theta$-UNE and the counterexample shows it is not necessary. Though due to Dafermos and Sparrow~\cite{dafermos1969traffic}, we know that for $\theta=1$ the variational inequality is the necessary and sufficient condition for $1$-PNE. This gives an alternative proof to the fact that $1$-UNE$=1$-PNE.
\end{remark}

For a fixed $\theta$, the flows satisfying Inequality~\eqref{eq:thetaVI} form a set. Call this set $\theta$-VI. 
The following lemmas characterize the price of anarchy for various flows under latency functions in class $\mc{L}$. We adopt the approach of Harks et al.~\cite{harks2007price} as it produces tighter PoA bounds for several latency functions compared to previous approaches~\cite{roughgarden2002selfish, correa2008geometric}. The result in~\cite{harks2007price} is for $\theta=1$ and here we state it for general $\theta$. 

We begin with a simple corollary to Lemma~\ref{lemm:UNEvVI} and omit the proof.
\begin{corollary}
For any multi-commodity instance $\mc{G}$ and any $\theta\geq 1$, the PoA values for the corresponding solution concepts are related as PoA($\theta$-PNE) $\leq $  PoA($\theta$-UNE) $\leq $  PoA($\theta$-VI).
\end{corollary}
  
We need the following definitions in order to bound PoA($\theta$-VI):
\begin{align*}
\omega(\mc{L}, \lambda) &= \sup_{\ell \in \mc{L}} \sup_{x,x'\geq 0}\frac{\left(\ell(x)-\lambda\ell(x')\right)x'}{x\ell(x)}.\\
\Lambda(\theta) &= \{\lambda\in \mathbb{R}^{+}: \omega(\mc{L}, \lambda)\leq  1/\theta \}.
\end{align*}

\begin{lemma}
For an instance $\mc{G}$ with latency functions in class $\mc{L}$, the PoA($\theta$-VI) is upper bounded by $ \inf_{\lambda \in \Lambda(\theta)} \theta\lambda(1-\theta\omega(\mc{L}, \lambda))^{-1}$. 
\end{lemma}
\begin{proof}
Let $\bm{x}$ be a $\theta$-VI flow satisfying Condition~\eqref{eq:thetaVI} and $\bm{y}\in SC_E$ be a socially optimal flow. Then, we have the following relations:
\begin{align*}
SC(\bm{x})= & \sum_e x_e\ell_e(x_e) \leq \theta\sum_e y_e\ell_e(x_e)\\
&\leq \theta\sum_e \left(y_e\ell_e(x_e)-\lambda y_e\ell_e(y_e) + \lambda y_e\ell_e(y_e)\right)\\
&\leq \theta  \omega(\mc{L}, \lambda) SC(\bm{x}) + \theta\lambda SC(\bm{y}).
\end{align*}
We obtain the desired bound by taking infimum over the set $\Lambda(\theta)$.
\end{proof}

\begin{example}
As an example consider the class of linear latency functions $\ell(x)= ax+b$.  For this class, we can obtain  $\omega(\mc{L},\lambda)\leq 1/4\lambda$ for $\lambda \geq 1$ and  $\omega(\mc{L},\lambda)>1$ otherwise. An upper bound on PoA($\theta$-VI) can be obtained through minimizing over the set $\Lambda(\theta)=\{ \lambda \geq \max\{1,\theta/4\}\}$. The exact bound obtained through this is $\max\{\theta^2, 4\theta/(4-\theta)\}$ and it matches the bounds given in~\cite{christodoulou2011performance}. Note that for $\theta=1$ it gives us the classical bound of $4/3$.
\end{example}

The following proposition further separates the envy free flows and the variational inequality characterization. 
\begin{proposition}
There exists a $1$-EF flow for which the variational inequality in \eqref{eq:thetaVI} does not hold for any bounded $\theta'$. Further, the PoA($1$-EF) is unbounded.
\end{proposition}
\begin{proof}
Consider the instance from Lemma~\ref{lemm:UNEvVI}, i.e. a Pigou network with demand $1$, top edge with latency $\ell_t(x) = \epsilon x$ and bottom edge with latency $\ell_b(x) = L$. The flow that routes $1$ unit through the bottom link (i.e., all the demand) is a $1$-EF flow but it does not satisfy Condition~\eqref{eq:thetaVI} for  any $\theta$, since under this flow the top link has cost 0.
For $L>2\epsilon$ the optimal flow routes all the flow through the upper link and thus the PoA of $1$-EF flows is $L/\epsilon$, which cannot be bounded as we can make $\epsilon$ arbitrarily small.
\end{proof}

\textbf{Price of Stability.} As discussed in the introduction, $\theta$-UNE and $\theta$-EF flows arise from the ability of a central planner to induce path flows in the network. The price of anarchy is motivated by the dynamics of users who can induce any worst case flow in the network under some given solution concept. In contrast, the price of stability is the quantity that is of special interest to the central planner, who wishes to induce the best (with respect to social cost) $\theta$-UNE or $\theta$-EF flow. Using already known techniques we can obtain bounds on the PoS but we defer this part to Section \ref{sec:designSubPoteFunc}, as these techniques will also be used to provide the central planner with good (with respect to social cost) $\theta$ fair flows, which is the scope of Section \ref{sec:design}.

\begin{lemma}\label{lemm:POSequality}
 For any multi-commodity network $\mc{G}$ and any $\theta\geq 1$, the PoS values of the corresponding flows are related as PoS($\theta$-EF) $\leq $  PoS($\theta$-UNE) $\leq $  PoS($\theta$-PNE). Moreover, there exists a network $\mc{G}$ such that for all $\theta\geq 1$ all the inequalities are tight.
\end{lemma}
\begin{proof}
The proof of the first part of the lemma follows due to Lemma~\ref{lemm:PNEvUNEvEF} and the fact that the infimum of a function over a set is less than or equal to the infimum of the same function over any subset of the set.

Consider the Pigou network with unit demand,  upper link having latency $\ell_u(x)= 1$ and  lower  link having latency $\ell_b(x) = x$. For this network and for any $\theta\geq 1$ the optimal $\theta$-PNE, $\theta$-UNE and $\theta$-EF are identical\footnote{For the special case of $\theta=1$, routing all the demand through the upper link is an optimal $1$-EF which is not a $1$-PNE or a $1$-UNE. Yet, the cost of this $1$-EF flow is the same as the $1$-PNE, $1$-UNE and $1$-EF flow that routes all the demand through the lower link. and given by $\max\{1/\theta, 0.5\}$ units of flow in the lower link and the remaining flow through the upper link. }
\end{proof}

\begin{remark}
It is important to emphasize that the upper bound guarantees of PoS encompass all possible networks under a given latency class. Whereas, for a particular network the achievable social cost under $\theta$-fairness can be better compared to the bound dictated by PoS. As an example, consider the latency function class of polynomials of degree at most $p$ where the best possible upper bound for PoS($\theta$-PNE) is $\left(\theta\left(1-\frac{p\theta^{1/p}}{(p+1)^{(1+1/p)}}\right)\right)^{-1}$ for $\theta< p+1$~\cite{christodoulou2011performance}. On the other hand, as shown by Defarmos et al.~\cite{dafermos1969traffic}, if all the latency functions are monomials of degree $p$, the $1$-PNE is the socially optimal flow.
\end{remark}

%% file: complexv2.tex
\section{Existence and Complexity}\label{sec:complex}
In this section we discuss the computational issues surrounding the three types of $\theta$ fair flows. The existence of Pure Nash equilibrium in nonatomic routing games guarantees the existence of any $\theta$ fair flow for $\theta\geq 1$. The next question would be whether we can compute $\theta$ fair flows with good social cost.  In particular, we consider the following problems:
\begin{enumerate}
	\item[(P1)] Find a $\theta$-EF path flow with the minimal social cost.
	\item[(P2)] Find a $\theta$-UNE path flow with the minimal social cost.
	\item[(P3)] Find a $\theta$-PNE edge flow with the minimal social cost.
\end{enumerate}
We show that for large $\theta$, the socially optimal flow is guaranteed to be contained in those $\theta$-flows, and hence the optimal $\theta$-flows be computed efficiently.  However, for small $\theta$, we will show that solving Problem~(P1) and Problem~(P2) is NP-hard, while it remains open whether Problem~(P3) can be computed efficiently.  
More precisely, for a latency class $\mc{L}$, this particular threshold is $\gamma(\mc{L})= \min\{\gamma: \ell^{*}(x)\leq \gamma \ell(x), \forall \ell\in \mc{L}, \forall x\geq 0\}$, where $\ell^{*}(x)= \ell(x)+x\ell'(x)$.  The main result of this section is given as follows:
\begin{theorem}
	For any multi commodity instance $\mc{G}$ 
	with latency functions in any class $\mc{L}$, there are polynomial time algorithms\footnotemark $\, $  for solving Problem~(P1)-(P3) 
	for $\theta \ge \gamma(\mc{L})$. 
	On the other hand, it is NP-hard to solve Problem~(P1) for $\theta \in [1, \gamma(\mc{L}))$ and Problem~(P2) for $\theta \in (1, \gamma(\mc{L}))$, for  arbitrary single commodity instances 
	with latency functions in an arbitrary class $\mc{L}$.
	\label{thm:main_hardness}
\end{theorem} 

\footnotetext{The existence of polynomial time algorithms for our problem depends on the assumption that we can minimize separable convex functions with linear constraints in polynomial time; numerical issues for convex optimization are discussed in \cite{hochbaum1990convex,  nemirovski2004interior} and are beyond the scope of our work.}

In the following sections, we first prove the first part of Theorem \ref{thm:main_hardness}. Right after we show that, for any $\theta$, from any $\theta$ fair flow we may get another $\theta$ fair flow, which uses only polynomially many paths. This, on the one hand, serves as a clarification that the difficulty of problems~(P1)-(P3) does not lie in the size of their solutions. On the other hand, it helps in showing that the decision version of these problems lies in NP, since for a YES instance, a non-deterministic machine will (non-deterministically) choose a path flow of polynomial size and in polynomial time check that it satisfies the conditions needed. Finally, the second part of Theorem \ref{thm:main_hardness} follows from an NP-hardness proof for a stronger version of the decision versions of problems~(P1) and (P2) (Theorem \ref{thm:12Hard}).

\subsection{When the Social Optimum is Guaranteed to be the Solution}\label{sec:complexity_so}
First, we show that Problems~(P1)-(P3) are easy for $\theta \ge \gamma(\mc{L})$ because the social optimum is the solution.
The following lemma, which is a direct extension of Theorem~4.2 in Correa et al~\cite{correa2007fast}, shows that any path decomposition of the socially optimal flow is a $\gamma(\mc{L})$ fair flow, provided that the latency functions are in class $\mc{L}$.  While in the proof of Theorem~4.2 in Correa et al~\cite{correa2007fast} they only conclude that the set of socially optimal path flows is $\gamma(\mc{L})$-EF, it is easy to see that the same argument holds for $\gamma(\mc{L})$-PNE.
\begin{lemma}(\cite{correa2007fast})
For a network $\mc{G}$ with latency functions  in class $\mc{L}$, any socially optimal path decomposition $o \in SO_p$ is $\gamma(\mc{L})$-PNE.
\end{lemma}
 
Since the social optimum can be computed using convex programming~\cite{roughgarden2002selfish}, it follows that 
Problems~(P1)-(P3) can be solved in polynomial time\footnotemark[6] for $\theta\ge\gamma(\mc{L})$.
\begin{proof}[Proof, first part of Theorem~\ref{thm:main_hardness}]
Note that given a path flow, which is $\gamma(\mc{L})$-PNE, it is $\gamma(\mc{L})$-UNE and $\gamma(\mc{L})$-EF as well.  This means that for all $\theta \geq \gamma(\mc{L})$, we can simply compute the socially optimal flow, and give any path decomposition as the $\theta$ fair flow. The socially optimal edge flow can be computed in time polynomial in the size of the network.  Further, a greedy path decomposition suffices. In the greedy algorithm, at every step we pick the current minimum path (among all commodities) and assign the maximum possible flow, under the social optimum, through this path. This can be computed in time $O(|\mc{K}|\times|E|)$.  Also, the output path flow can be represented with a sparse vector with $O(|\mc{K}|\times|E|)$ entries.
\end{proof}

\subsection{Existence of Polynomial-size Path Flow Solutions}
An observation to Problem~(P1) and (P2) is that the outputs of these two problems are path flow vectors, which are potentially of exponential size relative to the problem instances. 
In Section~\ref{sec:complexity_so} we showed a way to compute a path flow vector with polynomial support under the social optimum.  Here we ask whether we can do this for any edge flow.  In particular, we are interested in whether we can always find an answer to either Problem~(P1) or (P2) using only polynomially many paths.  If not, then there is no hope for us to find an efficient algorithm for these problems.  In this subsection, we show that the answer to this question is \emph{yes}.  To see this, we make a more general argument than Lemma~3.1 in Correa et al.~\cite{correa2007fast}, showing that given any path flow vector, we can always find another path flow assignment of polynomial support that preserves four important properties.
\begin{proposition}\label{lemma:correa1}
	Let $\bm{f}$ be a feasible flow for a multicommodity flow network with load-dependent edge latencies. Then, there
	exists another feasible flow $\bm{f}'$ such that 
	\begin{enumerate}
		\item $\bm{f}$ and $\bm{f}'$ have the same edge flow.
		\item The longest used path for commodity $k$ satisfies $\max_{\pi \in \mc{P}_u^k(\bm{f'})} l_{\pi}(\bm{f'}) \le \max_{\pi \in \mc{P}_u^k(\bm{f})} l_{\pi}(\bm{f})$.
		\item The shortest used path for commodity $k$ satisfies $\min_{\pi \in \mc{P}_u^k(\bm{f})} l_{\pi}(\bm{f}) \le \min_{\pi \in \mc{P}_u^k(\bm{f'})} l_{\pi}(\bm{f'})$.
		\item The flow $\bm{f}'$ uses at most $|E|$ paths for each source-sink pair.
	\end{enumerate}
\end{proposition}
\begin{remark}
The proof of this proposition directly follows the proof of Lemma~3.1 in Correa et al.~\cite{correa2007fast}, although our lemma statement is more general. (Their Lemma only states part 2 of our Lemma statement.)
\end{remark}

With this proposition, we can make the following argument that given an edge flow $\bm{x}$, if there is at least one $\theta$-EF or $\theta$-UNE path flow decomposition, then we can always find one with 
polynomial support:
\begin{lemma}\label{thm:npproblems}
	Given a $\theta$-EF path flow $\bm{f}_1$, there exists a $\theta$-EF path flow $\bm{f}'_1$ that uses at most $|E|$ paths for each source-sink pair and has the same edge flow as $\bm{f}_1$.  Similarly, given a $\theta$-UNE path flow $\bm{f}_2$, there exists a $\theta$-UNE path flow $\bm{f}'_2$ that uses at most $|E|$ paths for each source-sink pair and has the same edge flow as $\bm{f}_2$.
\end{lemma}
\begin{proof}
	For a $\theta$-EF path flow $\bm{f}_1$, by Proposition~\ref{lemma:correa1}, there exists a flow $\bm{f}'_1$ that has the same edge flow as $\bm{f}_1$ and the ratio of the longest used path to the shortest used path is bounded by
	$$
	\frac{\max_{\pi \in \mc{P}_u^k(\bm{f'}_1)} l_{\pi}(\bm{f'}_1)}{\min_{\pi \in \mc{P}_u^k(\bm{f'}_1)} l_{\pi}(\bm{f'}_1)} \le \frac{\max_{\pi \in \mc{P}_u^k(\bm{f}_1)} l_{\pi}(\bm{f}_1)}{\min_{\pi \in \mc{P}_u^k(\bm{f}_1)} l_{\pi}(\bm{f}_1)} \le \theta
	$$
	which indicates that $\bm{f}'$ is a $\theta$-EF path flow.  Similarly, given a $\theta$-UNE path flow $\bm{f}_2$, we can find a path flow $\bm{f}_2'$ that has the same edge flow as $\bm{f}_2$ and 
	$$
	\frac{\max_{\pi \in \mc{P}_u^k(\bm{f'}_2)} l_{\pi}(\bm{f'}_2)}{\min_{\pi \in \mc{P}^k} l_{\pi}(\bm{f'}_2)} \le \frac{\max_{\pi \in \mc{P}_u^k(\bm{f}_2)} l_{\pi}(\bm{f}_2)}{\min_{\pi \in \mc{P}^k} l_{\pi}(\bm{f}_2)} \le \theta
	$$
	from which we can conclude that $\bm{f}'_2$ is a $\theta$-UNE as well. 
\end{proof}
Now suppose $\bm{f}_1^*$ is the optimal solution to Problem~(P1).  According to Lemma~\ref{thm:npproblems}, we can see that there is an alternative path flow $\bm{f}_2^*$ that is also $\theta$-EF and has the same edge flow as $\bm{f}_1^*$.  Since the social cost only depends on the amount of the edge flow, $\bm{f}_1^*$ and $\bm{f}_2^*$ have the same social cost, from which we can conclude that $\bm{f}_2^*$ is an optimal solution to Problem~(P1) that uses only polynomially many paths.  A similar argument can be made for Problem~(P2) as well.

\subsection{Hardness Results}
In this section, we prove the second part of Theorem~\ref{thm:main_hardness} that it is NP-hard to solve Problem~(P1) and (P2) for small values of $\theta$.  More precisely, we consider the class of polynomial functions of degree at most $p$, which we denote as $\mc{L}_p$. We note that $\gamma(\mathcal{L}_p)=p+1$. We show that when the latency functions are in $\mc{L}_p$, then the related decision problems we state in Theorem~\ref{thm:12Hard} have polynomial-time reductions from the NP-complete problem PARTITION.  We state this result in the following theorem:


\begin{theorem}
For an arbitrary single commodity instace $\mc{G}$ 
with latency functions in class $\mc{L}_p$ for $p \ge 1$, it is NP-hard to
\begin{enumerate}
\item decide whether a socially optimal flow has a $\theta$-UNE path flow decomposition for $\theta \in (1, p+1)$.
\item decide whether a socially optimal flow has a $\theta'$-EF path flow decomposition for $\theta' \in [1, p+1)$.
\end{enumerate}
\label{thm:12Hard}
\end{theorem}

We state the following corollary that readily follows from Theorem~\ref{thm:12Hard}.
\begin{corollary}
For  any finite $\theta> 1$, it is NP-hard to find the optimal $\theta$-UNE or $\theta$-EF flow of an arbitrary instance $\mathcal{G}$.
\end{corollary}
\begin{proof}
For a given $\theta$ pick any $p\in \mathbb{N}:\theta<p+1$. Since $p+1=\gamma(\mathcal{L}_p)$, we may use  Theorem~\ref{thm:12Hard} to get the result.
\end{proof}

The proof of Theorem~\ref{thm:12Hard} is composed of two parts.  For the first part, we show the NP-hardness for $1.5$-UNE and $1$-EF path flow decompositions under the social optimum  in Lemma~\ref{lemm:corehardness}, based on the construction in Theorem~3.3 in Correa et al.~\cite{correa2007fast}.  Then, in the second part, we propose a novel way to generalize the construction to the entire range of $\theta$ and $\theta'$ specified in Theorem~\ref{thm:12Hard}.

\begin{lemma}\label{lemm:corehardness}
For single commodity instances with linear latency functions it is NP-hard to decide whether a social optimum flow has a $1.5$-UNE flow decomposition or a $1$-EF flow decomposition.
\end{lemma}
\begin{proof}
We consider the PARTITION problem, where we are given a set of $n$ positive integer numbers $q_1,\ldots, q_n$, and we need to decide  \emph{is there a subset $I \subset \{1,\ldots,n\}$ such that $\sum_{i\in I} q_i = \sum_{i \notin I}q_i$?}
 \begin{figure}[!htb]
 \centering
 \includegraphics[width=0.6\linewidth]{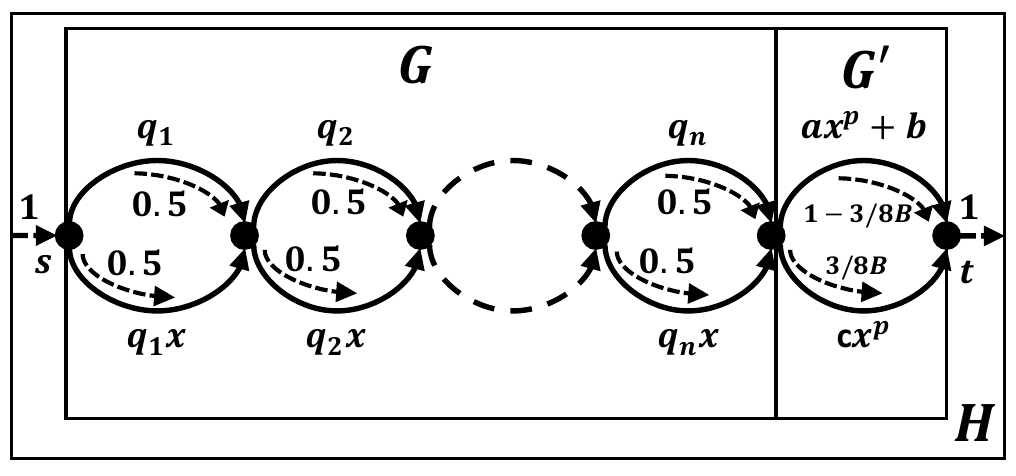}
 \caption{An instance of congestion game constructed from a given instance of PARTITION}
 \label{Fig:instForProg4}
 \end{figure}

Consider the two link parallel network with the top link $e_{u}$ having latency $\ell_u(x)=q$ and the bottom link $e_b$ having latency $\ell_b(x)=qx$. The demand between the source and the destination is $1$. 
The unique socially optimal flow splits the flow equally through the top and bottom link. Call this instance $G(q)$.

Given an instance of the PARTITION problem, $q_1,\ldots, q_n$, $\sum_{i=1}^{n} q_i=2B$, we now construct a single commodity network as the two link $n$ stage network $G$, as shown in Figure \ref{Fig:instForProg4}. In stage $i$ we connect $G(q_{i-1})$ to $G(q_{i})$ to the right for $i=2$ to $n$. A unit demand has to be routed from the source in $G(q_1)$ to the destination in $G(q_n)$.  For the graph $G$, the socially optimal flow $o$ routes $1/2$ flow through all top links and the remaining $1/2$ flow through each bottom link. We first observe that there is a one-to-one correspondence between the subsets $I\subseteq [n]$ and paths $p$ in $G$. Specifically, we can define the path corresponding to $I$ as $P_I=\left\{e_{u,i}: i\in I\right\} \cup \left\{e_{b,i}: i\notin I\right\}$. Further, the latency of the path is given by $\ell_I = \frac{1}{2}(\sum_{i\in [n]} q_i+\sum_{i\in I} q_i)$. 

In one direction, we observe that if the answer to the PARTITION problem is YES then there exists a subset $I^*$ such that $\sum_{i\in I^*} q_i = \sum_{i \notin I^*}q_i = B$. Consider the path flow under socially optimal flow $o$, with path $P_{I^*}$ carrying flow $1/2$ and path $P_{[n]\setminus I^*}$ carrying flow $1/2$.  The lengths of paths $P_{I^*}$ and $P_{[n]\setminus I^*}$ are both  equal to $\frac{3}{4}\sum_{i=1}^{n} q_i = 3B/2$. Whereas, the shortest path in the network is  $P_{\emptyset}$  with length $\frac{1}{2}\sum_{i=1}^{n} q_i = B$. Therefore, the socially optimal flow $o$ is a $3/2$-UNE flow and a $1$-EF flow, if $G$ comes from a YES instance of PARTITION.

In the other direction, we first observe that if a path $P_{I}$ under edge flow $o$ has length $3B/2 = \frac{3}{4}\sum_{i=1}^{n} q_i$, then $\sum_{i\in I} q_i = \frac{1}{2}\sum_{i\in [n]} q_i$. This implies the given answer to the PARTITION problem is YES. Now assuming $o$ is a $3/2$-UNE, there exists a path flow with the maximum used path of length less or equal to $\frac{3}{4}\sum_{i=1}^{n} q_i$. But the average length of any used path under $o$ is equal to $\frac{3}{4}\sum_{i=1}^{n} q_i$. This implies that all the paths in the path flow must have length $\frac{3}{4}\sum_{i=1}^{n} q_i$.  Next we assume that $o$ is a $1$-EF flow. This implies that there exists a path flow for which all the used paths have equal length. But then any used path under this decomposition has length $\frac{3}{4}\sum_{i=1}^{n} q_i$. Therefore, if $o$ is a $3/2$-UNE or a $1$-EF then the PARTITION instance corresponding to $G$ is a YES instance.   
\end{proof}
 
\begin{proof}[Proof of Theorem~\ref{thm:12Hard}]
Consider $\theta \in (1,p+1)$ for a UNE flow and $\theta' \in [1,p+1)$ for an EF flow. Given a PARTITION instance, let $G'$ be a two link parallel network with latency of the top link $\ell_{u,(n+1)}(x) = ax^p+b$ and bottom link latency $\ell_{d,(n+1)}(x) = cx^p$. We set $a=\frac{\alpha B}{(1-3/8B)^p}$, $b=\beta B(p+1)$, and $c=\frac{(\alpha+\beta)B}{(3/8B)^p}$, where $\alpha,\beta>0$ are some parameters to be determined later.

Using the fact that the social optimum is an equilibrium of the instance with latencies modified to $(\ell(x) + x\ell'_e(x))$, we  get that the socially optimal flow in network $G'$ is $\frac{3}{8B}$ through the bottom link and $(1-\frac{3}{8B})$ through the top link. We also get that at the social optimum the latency function satisfies the following condition:
$$
c\bigg(\frac{3}{8B}\bigg)^p = a\bigg(1-\frac{3}{8B}\bigg)^p + \frac{b}{p+1} < a\bigg(1-\frac{3}{8B}\bigg)^p + b
$$
From the latter, we can see that the top link has larger cost than the bottom link.  We then combine in series the network $G$ of Lemma \ref{lemm:corehardness} with the network $G'$ to obtain network $H$. The unique socially optimal flow in  network $H$ is the union of the two unique socially optimal flows in $G$ and $G'$. Recall the notation from Lemma \ref{lemm:corehardness}.

Assume the PARTITION problem admits a solution $I$. Consider the path decomposition in $H$:
\begin{enumerate}
	\item Path $p = P_{I}-e_{u, (n+1)}$ carries $1/2$ flow (note that $3/8B < 1/2$).
	\item Path $q = P_{I^c}-e_{u, (n+1)}$ carries $(1/2 - 3/8B)$ flow.
	\item Path $r = P_{I^c}-e_{d, (n+1)}$ carries $3/8B$ flow.
\end{enumerate}
We can see that the path $s = P_{\emptyset}-e_{d, (n+1)}$ is the shortest path, with latency $\ell_{s} = B + c(3/8B)^p = (\alpha+\beta+1)B$.  The longest used path $q$ has latency $\ell_{q} = 3B/2 + a(1-3/8B)^p + b = (\alpha+\beta+\beta p+3/2)B$.  Letting $c_1 = \frac{\ell_{q}}{ \ell_{s}}=\bigg(\frac{\alpha+\beta+\beta p+3/2}{\alpha+\beta+1}\bigg)$, the social optimum flow in $H$ is a $c_1$-UNE flow.

We next consider a different path flow for the EF setting. In this path flow:
\begin{enumerate}
	\item Path $s' = P_{[n]}-e_{d, (n+1)}$ carries $\frac{3}{8B}$ flow.
	\item Path $p = P_{I}-e_{u, (n+1)}$ carries $(1/2 - 3/8B)$ flow.
	\item Path $q = P_{I^c}-e_{u, (n+1)}$ carries $(1/2 - 3/8B)$ flow.
	\item Path $r'= P_{\emptyset}-e_{u, (n+1)}$ carries $\frac{3}{8B}$ flow.
\end{enumerate}
We claim that path $s'$ is the shortest path if $\beta p>1$ as
$$
\ell_{s'} = 2B + c(3/8B)^p = (\alpha+\beta+2)B <
(\alpha + \beta + \beta p + 1)B = B + a\bigg(1-\frac{3}{8B}\bigg)^p + b = \ell_{r'}<\ell_p=\ell_q.
$$
In this setting, the minimum ratio of longest `used' path and shortest `used' path is  $c_2 = \frac{\ell_{q}}{ \ell_{s'}}=\bigg(\frac{\alpha+\beta+\beta p+3/2}{\alpha+\beta+2}\bigg)$ and the socially optimal flow is a $c_2$-EF flow.

Next, we need to show that if the answer to PARTITION is NO then the socially optimal flow is neither a $c_1$-UNE flow nor a $c_2$-EF flow. For this we need to ensure that for all possible path flows under the social optimum, there exists at least one used path which is obtained by concatenating a `long' positive subpath in $G$ with the upper edge in $G'$. The following claim lower bounds the flow through the longest path in $G$ for any valid path decomposition. 

\begin{claim}\label{lemm:flowlower}
If the answer to PARTITION is NO then in the sub-network $G$ any path decomposition of the socially optimal flow $o$ routes at least $\frac{1}{2B}$ amount of flow through paths of length strictly greater than $\frac{3}{2}B$. 
\end{claim}
\begin{proof}
Recall that if the given instance for the PARTITION problem is a NO instance then there is no path under $o$ which has length exactly $3B/2$.
Fix any path decomposition for $o$ and let $\delta$ be the flow passing through the paths of length strictly greater than $\frac{3}{2}B$.   Also let $\ell$ be the maximum length among the set of paths strictly smaller than $\frac{3}{2}B$. As $q_i$'s are integers  and the given instance of PARTITION is a NO instance, it is easy to observe that $\ell \leq \frac{3}{2}B-\frac{1}{2}$. Also $\ell\geq B$. Moreover, if we route $(1-\delta)$ flow through a path of length $\ell$ and  $\delta$ flow through the path of maximum length $2B$, then the cost of this routing is greater or equal to the socially optimal cost. This implies,
$$\ell(1-\delta)+2\delta B \geq \frac{3}{2}B  \implies
\delta \geq \frac{3B/2-\ell}{2B-\ell} \geq \frac{3B/2-3B/2+1/2}{2B-B} \geq \frac{1}{2B}.$$
\end{proof}   

From the above claim we see that the longest used path $q$ has length strictly greater than $\ell_{q}$ as the bottom link under $o$ has flow $3/8B < 1/2B$. The shortest path in the network has length $\ell_{s}$ as in the YES case. If the PARTITION instance is a NO instance, the optimal flow $o$ is not a $c_1$-UNE. Moreover, for the EF flow the best path flow again contains the path $s'$ as the shortest path but now the longest path is strictly greater than $\ell_{q}$. So it is not a $c_2$-EF flow. 

All that is left to show is that there are appropriate values of $\alpha$ and $\beta$ which make $c_1 = \theta$ or $c_2 = \theta'$, for any $\theta\in (1, p+1)$, and for any $\theta'\in (1, p+1)$.  This can be shown by observing that:
\begin{align*}
	c_1=\frac{\alpha+\beta+\beta p + \frac{3}{2}}{1+\alpha+\beta}&=1+\frac{\frac{1}{2}+\beta p}{1+\alpha+\beta} & c_2=\frac{\alpha+\beta+\beta p + \frac{3}{2}}{2+\alpha+\beta}&=1+\frac{-\frac{1}{2}+\beta p}{2+\alpha+\beta}
\end{align*}   
Combining this with what we have shown in Lemma~\ref{lemm:corehardness} for $1$-EF flows completes the proof.
\end{proof}
\begin{proof}[Proof, second part of Theorem~\ref{thm:main_hardness}]
The proof follows by constructing a reduction from the decision problems specified in Theorem~\ref{thm:12Hard} and recalling that $\gamma(\mathcal{L}_p)=p+1$.  The answer to each of the decision problem in Theorem~\ref{thm:12Hard} is YES if and only if the solution to Problem~(P1) or (P2) is a social optimum, the cost of which is known in advance, by construction.
\end{proof}


For Problem~(P3), the proof technique in Theorem~\ref{thm:main_hardness} does not go through.  In fact, we show that the relevant decision problem related to Problem~(P3) is in P: 
\begin{enumerate}
\item[(P3')] Is there a socially optimal flow which is a $\theta$-PNE?
\end{enumerate}

To show that (P3') is in P, we first define an edge flow $\bm{x}$ to be \emph{acyclic} if for each commodity $k$, the subgraph $G_k$, induced by the edges $E_k(\bm{x}) = \{e: e\in E, x_e^k>0\} $ is a directed acyclic graph (DAG).
  
\begin{claim}
Given an instance of a multicommodity flow network $\mc{G}$ with standard latency functions, we can decide whether an `acyclic' edge flow $\bm{x}$ is in $\theta$-PNE in polynomial time.
\label{clm:Acyclic}
\end{claim}
\begin{proof}
We present the polynomial time algorithm which decides whether an `acylic' edge flow $\bm{x}$ is a $\theta$-PNE or not for some given $\theta$. For each commodity $k$ in $\mc{G}$, we construct the DAG induced by $E_k(\bm{x})$. Next, under the edge weights $w_e = \ell_e(x_e)$, we compute the costs of the shortest $(s_k, t_k)$ path in $G$ (call it $\ell_1$) and the longest $(s_k,t_k)$ path in $G_k$ (call it $\ell_2$).  Recall that shortest path computation and longest path computation in a DAG can both be done in polynomial time. Finally, we accept if $\ell_2 \leq \theta\ell_1$ and reject otherwise. 
\end{proof}

\begin{lemma}
Problem~(P3') can be solved in polynomial time.
\label{lemma:3Easy}
\end{lemma}
\begin{proof} 
We first claim that for any $k$, the set of edges that carry flow for commodity $k$ at the social optimum, $E_k(\bm{x}^*)$, has no positive loops.  This can be shown by contradiction.  Assume there is a positive loop in $E_k(\bm{x}^*)$, then, we can construct a new flow $\bm{x}'$ by removing some $\epsilon>0$ flow on the loop.  The flow $\bm{x}'$ can be kept feasible, and it has strictly smaller social cost due to the monotonicity and non-negativity of the latency functions, which contradicts the fact that $\bm{x}$ is the socially optimal flow.  Also, if there is a zero cost loop in $E_k(\bm{x}^*)$, we can safely remove the flow on that loop without changing the social cost. Therefore, the procedure in Lemma~\ref{clm:Acyclic} completes the proof. 
\end{proof}

%% file: design.tex
\section{Balanced Network Flow Design}\label{sec:design}
In this section, we first give an overview of the challenges of coming up with (designing) a flow that balances 
 the fairness and the social cost in the network. Drawing on results from previous sections, we highlight how the different solution concepts play important roles in balancing between the two objectives, minimizing social cost vs increasing fairness. We then present techniques to design edge flows with desired fairness level and low social cost. Finally, we show how introducing randomness can help in designing balanced network flows.

\smallskip\noindent\textbf{Central planner in flow design.} The task of a central planner has two components: 1) design an edge flow or path flow and 2) induce the designed flow by implementing a proper routing mechanism. 
 
The $\theta$-EF and $\theta$-UNE flows have the promise of coming up with a path flow which is both fair and has good social cost in a typical network. This is discussed through the example given in Section \ref{sec:exampleBalncedNetwork}. Unfortunately though, designing $\theta$-EF and $\theta$-UNE path flows is NP-hard as noted in  Section~\ref{sec:complex}. Moreover, there exist networks where  a $\theta$-EF flow with the lowest social cost is indeed a $\theta$-PNE, i.e.,  under a worst case framework the PoS($\theta$-EF) equals the PoS($\theta$-PNE) (recall Lemma \ref{lemm:POSequality}). Therefore, we focus on designing $\theta$-PNE flows that have both low social cost and low {\efu}. 
 
We discuss two main approaches of designing  such 
flows. One approach is the use of a modified potential function technique (Christodoulou et al.~\cite{christodoulou2011performance}) and the second approach is the bounded toll approach (Bonifacci et al.~\cite{bonifaci2011efficiency}). Here we clarify that we do not need to place tolls on the edges. We can virtually calculate the resulting edge flow and suggest to the central planner to induce this flow.

We further propose a randomized routing approach which tries to incentivize users to follow a given flow by making the average (over the randomness in the algorithm) latency of each user small. We show how the solution concepts play an important role in the variance reduction of the randomized routing.

\subsection{Improved balance using path flow: An example}\label{sec:exampleBalncedNetwork} 
 In this subsection, we argue that finding an appropriate path flow decomposition (as opposed to just specifying an edge flow) is critical to combining the goals of fairness and low social cost.  In fact, even in worst case examples, where a positive path may be very unfair (significantly longer than another path), there exist path flows that are completely fair. 


\smallskip\noindent\textbf{Worst Case Example.} Consider the instance depicted in Figure \ref{fig:badEgCasc} where we have $n$ stages of two parallel links connected in series. For any $i\in [n]$, in the $i$-th stage the top link $e_u(i)$ has a constant latency of $\ell_u = (2-\epsilon)$ and the bottom link $e_b(i)$ has latency function $\ell_b(x) = x$.  Here  $n\epsilon/2 \leq 1$. For a total demand of $1$, the social optimum in the network passes $(1-\epsilon/2)$ flow through the bottom link and the remaining flow through the top link for every stage. Whereas, the Nash equilibrium flow passes $1$ unit flow through the bottom link. The  $maxpath/minpath$ ratio  of the SO flow is equal to $2$, which is the worst possible under linear latencies. The same conclusion holds for any class of latency functions by replacing $x$ with $\ell(x)$ and $(2-\epsilon)$ with $\ell^*(d)-\epsilon$. Here $d$ is total demand, $\ell^*(x)= \ell(x)+x\ell'(x)$ and $\gamma(\mc{L})= \ell^*(d)/\ell(d)$. This was noted in Correa et al.~\cite{correa2004computational}.   


\smallskip\noindent\textbf{Balanced UNE and EF flow.} This worst case example admits an almost balanced path flow under the SO flow.
Let $p_i$ be a path that uses the top edge from $i$-th stage and the bottom edge in the other stages, for each $i\in [n]$ and let $p_0$ be the path using only bottom edges from each stage. We consider the path flow decomposition of the social optimum where $\epsilon/2$ flow passes through path $p_i$,  for all $i\in [n]$, and the remaining flow, if any, passes through path $p_0$. We can easily check this is a valid path flow and  a 
$\left(1+ \frac{1}{n}\right)$-UNE. 
Note that the term $\left(1+ \frac{1}{n}\right)$ approaches $1$ as $n$ becomes 
large. Moreover, for $n\epsilon/2=1$, the path $p_0$ is not used and this flow is indeed a $1$-EF flow. This shows that even worst case examples in terms of edge flow can produce a path flow with near optimal results.

\smallskip\noindent\textbf{Balanced PNE.} Moreover, it is also possible to bring down the {\efu} by compromising on the social cost slightly. Let $k$ stages out of a total of $n$ stages be in SO locally and the other $(n-k)$ stages be in NE locally. This flow is a 
$\left(\frac{n+k(1-\epsilon)}{n-k\epsilon/2}\right)$-PNE, whereas the social cost is $\left(n -k\epsilon^2/4\right)$. This presents us with a complete spectrum of balanced flows.  
\begin{figure}[!htb]
\centerline{\includegraphics[width=0.5\linewidth]{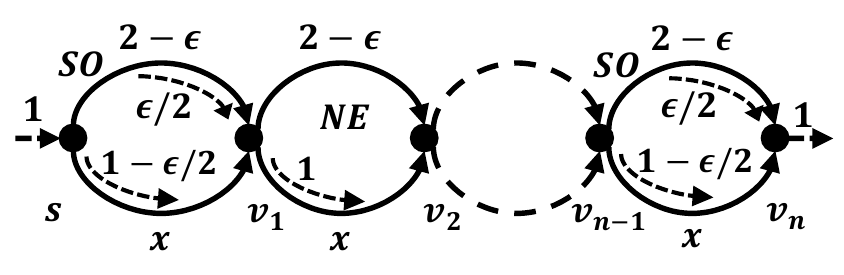}}
\caption{Improved Balance: Example}
\label{fig:badEgCasc}
\end{figure}

\subsection{Edge Flow Design based on Modified Potential Functions}\label{sec:designSubPoteFunc}
Consider the flow minimizing a modified potential, specified by $\bm{x}^* = \argmin_{\bm{x}\in \mc{D}_E} \sum_e \int_0^{x_e}\phi_e(t)dt$. 
From Theorem 4 in~\cite{christodoulou2011performance}  we know  that if for all edges $e$ and for all $x\geq 0$, this modified potential $\phi_e(x)$ satisfies $ \ell_e(x)/\theta\leq \phi_e(x)\leq \ell_e(x)$, then $\bm{x}^*$ is a $\theta$-PNE. Further, we can bound the inefficiency of the flow $\bm{x}^*$ as the PoA($1$-PNE) under the modified potential functions. This gives us an upper bound for the PoS($\theta$-PNE). From the inclusion of the $\theta$ flows in Lemma~\ref{lemm:PNEvUNEvEF}, this is the upper bound for both PoS($\theta$-UNE) and PoS($\theta$-EF).

\input{thetaNashNew.tex}

\subsection{Edge Flow Design based on Bounded Tolls} 
In the routing games literature a natural way to enforce the socially optimal flows has been to place tolls on the edges. Among the many variations of this problem a practical one is to consider tolls that are bounded on every edge, which may not attain the social optimum but reduce the social cost of the best equilibrium, namely they lower the price of stability. Bonifacci et al. in~\cite{bonifaci2011efficiency} considered a version of bounded tolls where each edge $e$ has an upper bound on its toll given by $\epsilon \ell_e(x)$ for all $x\geq 0$. Call this an $\epsilon$-bounded toll. The following lemma shows how we can use this idea in the context of computing good $\theta$-PNE flows.

\begin{lemma}
The $1$-PNE under tolled latency functions, with  $\epsilon$-bounded tolls, is a $(1+\epsilon)$-PNE under the original latency functions.
\end{lemma}
\begin{proof}
Let $\hat{\bm{x}}$  be a $1$-PNE under the tolled latency functions $\{\tau_e(\cdot)\}_{e\in E}$. For any commodity $k$ let $p$ be a positive path under this flow and $q$ be any other path. The tolled cost of $p$ is less or equal to the tolled cost of $q$, i.e. $\sum_{e\in p}\tau_e(\hat{x}_e)\leq \sum_{e\in q}\tau_e(\hat{x}_e)$. As the tolls are $\epsilon$-bounded we have $\tau_e(x)\in [0,\epsilon\ell_e(x)]$, hence the original latencies of the two paths satisfy the bound $\ell_p(\hat{\bm{x}}) \leq (1+\epsilon) \ell_q(\bm{x}^b)$. Therefore, $\hat{\bm{x}}$ is a $(1+\epsilon)$-PNE under the original latency functions $\ell_e(\cdot)$.
\end{proof}
    
This presents us with another strategy where we use latency functions with $(\theta-1)$ bounded tolls and obtain the resulting equilibrium as a $\theta$-PNE that also has good social cost. For example, 
Bonifacci et al.~\cite{bonifaci2011efficiency} considered tolls of the form $\min \{x\ell'_e(x), \epsilon \ell_e(x)\}$ and showed that for polynomial latency functions the cost of the tolled equilibrium flow is upper bounded by PoS($\theta$-PNE) of instances with polynomial latency functions. As another example, Fotakis et al.~\cite{fotakis2015improving}, in a setting technically similar  to that of bounded tolls,  provide, for series parallel graphs, an upper bound   on the PoS for general latencies in class $\mathcal{D}$, which, in our setting, equals to 
$\max\Big\{1,\frac{1}{1-\beta_\theta(\mathcal{D})}\Big\}$, with
$\beta_\theta(\mathcal{D})=\sup_{\ell\in \mathcal{D}, x\geq
y\geq0}\frac{y(\ell(x)-\ell(y))-(\theta-1)(x-y)\ell(x)}{x\ell(x)}$.

\subsection{Randomized Routing}
In this section, we introduce the idea of randomized routing in traffic networks as a way of implementing a given flow.
The design of deterministic routes in a given network presents us with the dichotomy that flows with good social cost have inherent unfairness. Therefore, in many cases, if we want to improve the social cost we must assign some user a long path which causes dissatisfaction on her part. We can get out of this seemingly unavoidable situation through the use of randomization in route assignment.
In randomized routing (RR) the central planner tries to induce a specific path flow $\bm{f}$ in the network, by assigning each user randomly to some route. For any commodity $k$ and a user with this commodity, the randomized routing assigns the user to a path $p\in \mc{P}^k$ with probability $f_p/ d_k$. In what follows, we formalize the routing process within a distributed framework.
 
\textbf{User Identity.}  Let us represent each infinitesimal user with commodity $k$ using a real number that takes value from $[0,1]$. Therefore, we can label each user by a tuple $(k,id)$ where $k$ represents the commodity and $id\in[0,1]$. 

\textbf{Path ordering.} Given a path flow $\bm{f}$ we can impose an arbitrary order on the set of `used' paths for each commodity $k$, i.e. $\left(p_k(i) : i \in \{1, \dots, |\mc{P}^k_{u}|\}\right)$. Here the ordering means the ids in range $\Big[\frac{\sum_{j=1}^{(i-1)} f_{p_k(j)}}{d_k}, \frac{\sum_{j=1}^{i} f_{p_k(j)}}{d_k}\Big)$ are assigned to path $p_k(i)$. 

\textbf{Hash functions.} Under this ordering with the given flow $\bm{f}$ we can define a hash function $h_x: [0,1]\times\mc{K}\rightarrow \mathbb{N}$ as $h_x(k,id)= p_k(i^*)$ where $i^* = \argmax\left\{i:  \text{frac}(id+x) \leq \sum_{j=1}^{i} f_{p_k(j)}/ d_k \right\}$. Here $\text{frac}(x)$ gives the fractional part of $x$. This hash function divides the real line $[0,1]$ into intervals and assigns the $i$-th interval to the $i$-th path. The length of the $i$-th interval is proportional to the flow in $i$-th path. Given a path flow with polynomially many `used' paths it is possible to compute this hash function efficiently (with some quantization).  
 
\textbf{Algorithm.} The randomized routing can be implemented as a distributed system as presented below.
\begin{enumerate}
\item The central planner picks $X$ uniformly at random from $[0,1]$.
\item Given a path flow $\bm{f}$ the central planner computes the hash function $h_X$ and broadcasts it. 
\item Upon receiving the hash function $h_X$, a user $(k,id)$ chooses the path $p = h_X(k,id)$.
\end{enumerate}

\textbf{Performance.} In the next lemma, we characterize some properties of the randomized routing when the central planner induces a  $\theta$-UNE or $\theta$-EF flow.
\begin{theorem}
A randomized routing (RR) under a flow $\bm{f}\in \theta$-UNE  $\cup$  $\theta$-EF has the following properties:
\begin{enumerate}
\item RR induces the original flow $\bm{f}$.
\item Each user with commodity $k$ experiences the latency $\bar{\ell}_k = (\sum_{p\in \mc{P}_{u}^k} f_p \ell_p(\bm{f}))/d_k$ in expectation. Therefore, the RR produces a $1$-EF flow in expectation. 
\item The standard deviation of the latency seen by a typical user with commodity $k$ is upper bounded by $\frac{(\theta-1)}{4\sqrt{\theta}}\bar{\ell}_k $.
\end{enumerate}
\label{thm:RR}
\end{theorem}
\begin{proof}
The RR performs a randomized rotation of the user ids and then assigns the paths according to the new ids. Here the change of the ordering does not affect the amount of flow any path $p$ is assigned to, i.e. path $p$ is assigned exactly $f_p$ amount of flow. The randomized routing induces the flow $\bm{f}$. 

For some user $(k,id)$, the randomized rotation creates the new id $(k,y)$ where $y = \text{frac}(id+X)$. For $X\sim U([0,1])$ we get $y\sim U([0,1])$ from basic probability theory. This implies that the probability that user $(k,id)$ is assigned a path $p_k(i)$ is $f_{p_k(i)}$ for any $i$. Therefore, the expected latency for user $(k,id)$ is $\bar{\ell}_k = (\sum_{p\in \mc{P}_{u}^k} f_p \ell_p(\bm{f}))/d_k$.

Let the maximum path in $\mc{P}_u^k$ has length $L_k$ and the minimum path has length $l_k$, for each $k$. From the Bhatia-Davis bound~\cite{bhatia2000better} on the variance of a random variable we get that the variance of the latency of any user with commodity $k$ is bounded from above by $(L_k - \bar{\ell}_k)(\bar{\ell}_k - l_k)$. Further, we know that $L_k \leq \theta l_k$ as $\bm{f}\in \theta$-UNE  $\cup$  $\theta$-EF. Through simple algebra we obtain the upper bound on the variance as $\frac{(\theta-1)^2}{16\theta}{\bar{\ell}}_k^2$.  
\end{proof}
  
\begin{remark}
In Section~\ref{sec:complex}, we showed that the computation of the optimal $\theta$-UNE or $\theta$-EF flows is NP-hard. 
For this reason, in this section we turned to computing a $\theta$-PNE with low social cost and then proceeded to compute a path flow from this edge flow. Any path flow we computed here is $\theta$-UNE or $\theta$-EF by Lemma~\ref{lemm:PNEvUNEvEF}. There is also the possibility of improving $\theta$ by avoiding a long `positive' path. A candidate method for this improvement can be the technique mentioned in Proposition~\ref{lemma:correa1}, whereby we start from a path flow with more than $|A|$ used paths, and gradually eliminate current longest and shortest paths to make the resulting flow more fair.  
\end{remark}

%% file: thetaNashNew.tex
The choice of proper functions $\phi_e(\cdot)$ combined with the $\lambda$-$\mu$  smoothness  framework for latency functions~\cite{roughgarden2015intrinsic} enables us to strictly improve the social cost compared to that of $1$-PNE, thus bounding PoS($\theta$-PNE) away from PoA($1$-PNE). Following the ideas in~\cite{christodoulou2011performance}, we present a structured method to find good modified latency functions and extend the PoS bounds to the class of $M/M/1$ latency functions, which is commonly used in modeling congestion networks.

Given a standard latency function $\ell(\cdot)$ and a range $\mc{R}$, consider the class of functions $\mc{L}(\ell, \mc{R}) =\{ \phi(\cdot): \phi(\cdot) \text{ is standard}, \ell(x)/\theta \le \phi(x) \le l(x),~\forall x \in \mc{R}\}$.
Further, given a multi-commodity network $\mc{G}$ with total demand $d_{tot}$, define the class of new potential functions,
\begin{flalign}
\label{eq:newPotential}
\bm{\Phi}(\mc{G})=\left\{\sum_{e \in E} \int_{x=0}^{x_e}\phi_e(x) dx: \phi_e(\cdot)\in  \mc{L}(\ell, [0, d_{tot}])\right\}.
\end{flalign}

The following result from \cite{christodoulou2011performance} characterizes the $\theta$-PNE in $\mc{G}$.
\begin{proposition}\label{thm:NewNE}
Given a multi-commodity network $\mc{G}$, a feasible flow $\bm{x}$ is a $\theta$-PNE if it minimizes some potential function $\Phi(\bm{x}) \in \bm{\Phi}(\mc{G})$.
\end{proposition}

\smallskip\noindent\textbf{ PoS Upper Bounds for composite functions.}
Consider the class of latency functions represented as $\ell(x)= \sum_i a_i\ell_i(x)$ where $a_i\geq 0$ for all $i$. Let the total demand in the network be $d=\sum_k d_k$.  We can find an upper bound for PoS through the following procedure:

\begin{itemize}
\item[1.] For each $i$, guess a suitable form of function $\phi_i(x, \psi_i)$, where $\psi_i$ is an appropriately chosen parameter. Represent $\phi(x)= \sum_i \xi_i a_i\phi_i(x, \psi_i)$  for $\xi_i \in [1/\theta,1]$.
\item[2.] For each $i$, obtain the set 
\begin{align*}
\Psi_i(\theta,\xi_i) = \{\psi: \xi_i\phi_i(x, \psi) \in [\ell_i(x)/\theta, \ell_i(x)],~\forall x \in [0,d] \}.
\end{align*}
\item[3.] For each $i$, obtain the set 
\begin{align*}
\Lambda_i(\psi) = \{(\alpha, \beta): y \phi_i(x, \psi) \leq \alpha x \phi_i(x, \psi)  + \beta y \phi_i(y, \psi),~\forall x,y \in [0,d]\}.
\end{align*}
\item[4.] Solve the following optimization problem,
\begin{align*}
\label{eq:designPoS}
PoS(\theta)= \min \left\{\frac{\beta_p}{1-\alpha_p}:  \frac{1-\alpha_p}{1-\alpha_i} \leq \xi_i \leq \frac{\beta_p}{\beta_i},(\alpha_i,\beta_i) \in \Lambda_i(\psi_i), \psi_i \in \Psi_i(\theta,\xi_i), \xi_i \in [1/\theta,1], ~\forall i\right\}.
\end{align*}
\end{itemize}

\smallskip\noindent\textbf{$\bm{M/M/1}$ Delay functions.}
Consider latency functions in the class $\mathcal{D} =\{1/(u-x): u\geq u_{min}\}$, where $u$ is the capacity of the link and $x$ is the flow through the link. The term $u_{min}$ refers to the minimum capacity in the latency class. Further for each function the maximum load is given as $\rho = d/u$ and therefore, $\rho_{max} = d/ u_{min} < 1$ denotes the maximum possible load over the entire class. This is the class of $M/M/1$ delay functions which plays an important role in modeling congestion networks. 
  
\begin{lemma}
The PoS for the latency functions in class $\mathcal{D}$ for $\theta$-PNE,  $\theta \geq  1$, is upper bounded as
\begin{align*}
 \text{PoS}(\theta\text{-PNE};\mathcal{D})\leq \frac{1}{2}\left( 1+ \frac{1}{\sqrt{1- \rho_{\max}(\theta)}}\right),
\end{align*}
where $\rho_{\max}(\theta) = \max \{0, 1-\theta(1-\rho_{\max})\}$. Moreover, if $\theta \geq 1/(1-\rho_{max})$, the PoS of the network becomes $1$. 
\label{lemm:MM1}
\end{lemma}
\begin{proof}
Step 1: Consider the original function $\ell(x;u) = 1/(u-x)$ and the new functions $\phi(x;a,u) = 1/(u-ax)$ for some $a\in \mathbb{R}_{+}$.  Call the class of modified functions, $\mc{D}_{a} = \{\phi(x;a,u): u \geq u_{min}\}$.  

Step 2: Define the set $\Psi(\theta)= \{a:  a\rho\in [\max\{0,(1-\theta(1-\rho))\}, 1 ]\}.$ 

For $a\in \Psi(\theta)$, we have, $ \ell(x;u)/\theta \leq \phi(x;a,u)\leq \ell(x;u) $, for all $u$ and for all  $0\leq x\leq d$. 

The solution to the program that minimizes \eqref{eq:newPotential} is the Nash equilibrium under the latency functions $\phi_e(x;a_e,u_e)$. But, from Proposition~\ref{thm:NewNE}, the solution is a $\theta$-PNE for the original system with functions $\ell_e(x;u_e)$, if all $a_e\in \Psi(\theta)$. Therefore, the price of anarchy (PoA) under latency functions of the class $\mc{D}_{a}$ for all $a\in \Psi(\theta)$, gives upper bounds for the PoS for $\theta$-PNE. 

Step 3: 
The class of functions $\mc{D}_{\phi}$ assumes the same form but changes the maximum load on the system compared to $\mc{D}$. Specifically, for any fixed $a\in  \Psi(\theta) $, we can have the following relation true for $(\alpha,\beta)\in \Lambda(a)$,
\begin{align*}
y\phi(x; a, u) \le \alpha x\phi(x; a,u) + \beta y\phi(y; a, u),& \qquad \forall x,y \in [0,d].
\end{align*} 
Here through some basic calculus we can find out that if the inequality holds on the boundary of $[0,d]^2$ then it holds for the entire region. We obtain that for the boundary $y=0$ the inequality  is always true and for the boundaries $x = 0$ and  $y = d$, the  inequality holds for $\beta \geq 1$ (necessary and sufficient). Further, for $x=d$ the condition,  $4a\rho \alpha  \geq(1+a\rho\alpha -  \beta(1-a\rho))^2$, is  both necessary and sufficient. Therefore, we have $$\Lambda(a) = \{(\alpha,\beta): \alpha\in[0,1),  \beta\geq 1,  4a\rho \alpha  - (1+a\rho\alpha - \beta(1-a\rho))^2\geq 0 \}.$$

Step 4: We can now get the PoS upper bound after minimizing $$\min \{\beta/(1-\alpha): (\alpha,\beta) \in \Lambda(a), a \in \Psi(\theta)\}.$$ 

After the optimization, we get that the minimum is $\frac{1}{2}\left( 1+ \frac{1}{\sqrt{1- \rho(\theta)}}\right)$, where $\rho(\theta) = \max \{0, 1-\theta(1-\rho)\}$ and the choice of $a= \rho(\theta)/\rho$. Finally, taking maximum over all possible $\rho$ we obtain  $\text{PoS}(\theta\text{-PNE};\mathcal{D})\leq \frac{1}{2}\left( 1+ \frac{1}{\sqrt{1- \rho_{\max}(\theta)}}\right)$, where $\rho_{\max}(\theta) = \max \{0, 1-\theta(1-\rho_{\max})\}$.

\end{proof}

%% file: conclusion.tex
\section{Conclusion}
In this article, we investigated the specific roles played by edge flows and path flows in achieving fairness in traffic routing without compromising on the social welfare too much. To this end we differentiated between `used' paths and `positive' paths. The former relates to paths with non zero flow under a given path flow, while the latter relates to paths with non zero flow on each edge under a given edge flow. The understanding of these two new flows led us to new solution concepts which generalize the classic Nash  equilibrium in routing games. Specifically, we defined positive Nash equilibrium (PNE) as an edge flow where the length of any `positive' path for any commodity is less than or equal to the length of any path of the same commodity.  Substituting `positive' paths with `used' paths in the definition of PNE gives us the concept of used Nash equilibrium (UNE). Relaxing the conditions further, we obtained envy free (EF) flows where for each commodity all `used' paths have equal length (in particular, this concept allows for the existence of unused paths of shorter length). In the spirit of approximate Nash equilibria, we considered the approximate versions of these solution concepts, $\theta$-PNE, $\theta$-UNE and $\theta$-EF, for some constant $\theta\geq 1$. Each of them yielded as a by-product a $\theta$-fair flow, under the fairness definition in Correa et al.~\cite{correa2007fast}.  However, we note that depending on the users' affinity towards selfishness and their knowledge of the network congestion one of these solution concepts might be more relevant than the others.  
  
We explored the interrelations among these flows building a reasonably complete hierarchy among them. Further, using the well developed framework of variational inequalities, we analyzed the price of anarchy (PoA) and price of stability (PoS) of $\theta$-PNE, $\theta$-UNE and $\theta$-EF flows. The results for PoA and PoS successfully encapsulate all possible instances of a multi-commodity network with latency functions from a given class. However, for a particular instance they fail to quantify the social cost efficiently.  
We then investigated the computational complexity of
finding $\theta$-PNE, $\theta$-UNE and $\theta$-EF flows with lowest social cost.  We proved that finding the $\theta$-UNE or $\theta$-EF with lowest social cost is NP-hard and remarked on the complexity of finding optimal $\theta$-PNE flows, which remains open.  
To circumvent these negative results, we then connected two existing approaches for related problems, namely, 1) bounded tolls~\cite{bonifaci2011efficiency} and  2) modified potential function~\cite{christodoulou2011performance}, to design edge flows which are $\theta$-PNE and have low social cost. Finally, we proposed a randomized routing where a central planner assigns a route to a user randomly and guarantees that the induced flow is $1$-EF flow `in expectation'. In fact, this technique can effectively induce the socially optimal flow as a $1$-EF flow `in expectation' but  possibly with large variance. The natural way to bound the variance is using a $\theta$-EF flow. Unfortunately, as of now we can not compute a $\theta$-EF or $\theta$-UNE flow directly and have to take recourse to $\theta$-PNE. We prove bounds on the variance of the length of the route assigned to the user by the randomized routing when starting from one of $\theta$-PNE, $\theta$-UNE or $\theta$-PNE.
         
We leave the following open problems and future directions, which can be used to design balanced flows with good social cost and fairness `in expectation': 
\begin{enumerate}
\item What is the computational complexity of (P3), i.e. calculating a $\theta$-PNE with the lowest social cost? Lemma~\ref{lemma:3Easy} shows how our technique fails to resolve the hardness in this case.
\item Can we design approximation algorithms to generate $\theta$-UNE or $\theta$-EF with near optimal social cost?
\item How can we formalize the notion of fairness in the presence of randomization in routing algorithms? 
\end{enumerate}

%% file: acknowledgements.tex
\paragraph{Acknowledgements} This work was supported in part by NSF grants CCF 1216103 and 1331863, an NSF CAREER Award and a Google Faculty Research Award.  Part of the research was performed while a subset of the authors were at the Simons Institute in Berkeley, CA in Fall 2015.  